\documentclass[format=acmsmall, dvipsnames]{acmart}
\pdfoutput=1
\usepackage{bm}
\usepackage{dsfont}
\usepackage{algorithm}
\usepackage{algorithmic}
\usepackage{cuted}
 
\usepackage[export]{adjustbox}
\makeatletter
\g@addto@macro{\UrlBreaks}{\UrlOrds}
\makeatother
\usepackage{breqn}
\usepackage{enumitem}
\setlist[itemize]{leftmargin=*}
\usepackage[normalem]{ulem}
\usepackage{xcolor}

\newcommand{\change}[1]{\textcolor{black}{#1}}

\newcommand{\oo}{{\rm o}}

\newcommand{\be}{\begin{eqnarray}}

\newcommand{\ee}{\end{eqnarray}}
\newcommand{\ben}{\begin{eqnarray*}}
	\newcommand{\een}{\end{eqnarray*}}
\newcommand{\bfl}{\begin{flalign*}}
	\newcommand{\efl}{\end{flalign*}}

\newcommand{\calJ}{{\mathcal J}}

\setcopyright{none}
\renewcommand\footnotetextcopyrightpermission[1]{} 
\begin{document}



\title{On Non-Preemptive VM Scheduling in the Cloud}

\author{Konstantinos Psychas, and Javad Ghaderi}
\affiliation{%
	\institution{Columbia University}
}
\thanks{This work was supported by NSF Grant CNS-1652115}

\begin{abstract}
We study the problem of scheduling VMs (Virtual Machines)
in a distributed server platform,
motivated by cloud computing applications.
The VMs arrive dynamically over time to the system,
and require a certain amount of resources
(e.g. memory, CPU, etc) for the duration of their service.
To avoid costly preemptions, we consider non-preemptive scheduling:
Each VM has to be assigned to a server
which has enough residual capacity
to accommodate it, and once a VM is assigned to a server,
its service \textit{cannot} be disrupted (preempted).
Prior approaches to this problem either have high complexity,
require synchronization among the servers, or yield queue sizes/delays which are excessively large.
We propose a non-preemptive scheduling algorithm that resolves these issues. In general, given an approximation algorithm to Knapsack with approximation ratio $r$, our scheduling algorithm can provide $r\beta$ fraction of the throughput region for $\beta < r$. In the special case of a greedy approximation algorithm to Knapsack, we further show that this condition can be relaxed to $\beta<1$. The parameters $\beta$ and $r$ can be tuned to provide a tradeoff between achievable throughput, delay, and computational complexity of the scheduling algorithm. 
Finally extensive simulation results using both synthetic and real traffic traces are presented to verify the performance of our algorithm.
\end{abstract}
%
\acmVolume{} \acmNumber{} \acmArticle{} 
\acmMonth{6} 
\acmYear{2018} \acmPrice{} \acmDOI{}

\maketitle

%


\section{Introduction}
There has been an enormous momentum recently in moving storage, computing, and
various services to the cloud. By using cloud, clients no longer require to
install and maintain their own infrastructure  and  can  instead  use  massive
cloud  computing resources on demand (for example, Expedia~\cite{expedia} and
Netflix are hosted on Amazon's cloud service~\cite{AEC2}). Clients can procure
Virtual Machines (VMs) with specific configurations of CPU, memory, disk, and
networking in the cloud. In a more complex scenario, clients can put together
an entire service by procuring and composing VMs with specific
capabilities~\cite{aws-datapipe,john2013research}.

The datacenter is a distributed server platform, consisting of a large
number of servers. The key challenge for the cloud operator is to efficiently
support a wide range of applications on their physical platform. Recent studies estimate
in many large datacenters the average server utilization to be 6 to 12\%
(see~\cite{glanz2012power} and references therein). At such low utilizations,
VMs can be potentially concentrated onto a smaller number of servers, and many
of the unused servers can be turned off (to save energy) or utilized to increase
the number of VMs that can be simultaneously supported by the system (to maximize
throughput and reduce delay). For instance, suppose a CPU-intensive VM, a
disk-intensive VM, and a memory-intensive VM are located on three individual
servers, we can pack these VMs in a single server to fully utilize the server's
resources along CPU, disk I/O, and memory. However, finding the right packing of
VMs is not easy due to two reasons: first, the cloud workload is a priori unknown
and will likely be variable over both time and space; and second, finding the right
packing even in the case that the workload is known is a hard combinatorial problem.

In this paper, we consider a distributed server platform, consisting of
possibly a large number of servers. The servers could be inhomogeneous
in terms of their capacity (e.g. CPU, memory, storage).
As an abstraction in our model, VM is simply a
multi-dimensional object (vector of resource requirements)
that \textit{cannot} be fragmented among the servers. The VMs of various
types arrive dynamically over time. Once a VM arrives, it is queued and
later served by one of the servers that has sufficient remaining capacity
to serve it. Once the service is completed, the VM departs from the server and
releases the resources.

We consider non-preemptive scheduling, i.e., once a VM starts getting service,
its ongoing service cannot be preempted (interrupted). This is because preemptions
require storing the state of preempted VMs and recovering them at a later time,
which are operationally costly and can also affect the latency~\cite{Dargie2014}.
Admittedly there are scenarios where preemptions could be actually
necessary/useful, e.g. for maintenance, low cost pricing, energy
saving~\cite{PreemtibleVMs,clark2005live, Lin2013}, or for resource allocation in
long-running services (e.g., a long-running VM where the cost of one-time
preemption can be amortized over the VM's life time). In this paper, we focus on
non-preemptive scheduling, and postpone the preemption cost modeling to a separate
work.

We are interested in scalable non-preemptive scheduling algorithms
that can provide high throughput and low delay.
To maintain scalability, we would like the scheduling decisions
to be made by the servers individually in a distributed manner, without
the need for coordination among the servers. In this work, we propose an
algorithm to meet these objectives and will characterize its theoretical
performance. Further extensions are also discussed to make the algorithms
more applicable to realistic settings.

We would like to emphasize that although we use the term VM, our model provides clean abstractions and algorithms that can be applied to
other applications as well. For example, in scheduling tasks in data-parallel clusters, tasks can be viewed as VMs in our model (multi-dimensional objects) with diverse resource requirements
(CPU, memory, storage, etc)~\cite{grandl2014multi}.

\subsection{Motivations and Challenges}\label{sec-related}
Consider a large-scale server system with a finite number of VM types.
At any time, each server could operate in one of many possible configurations,
where each configuration is a way of packing various number of VM types in
the server subject to its capacity. As VMs arrive and depart
over time, the configuration of servers may need to change
appropriately in order to schedule the VMs waiting to get service. To avoid costly preemptions, the
configuration change has to be non-preemptive. For example, suppose there
are only two VM types, if the server configuration is $(2, 2)$ (i.e., it is
currently serving $2$ VMs of type $1$ and $2$ VMs of type $2$), it cannot
suddenly transition to $(0,4)$ (i.e., serving $4$ VMs of type
$2$, and $0$ VMs of type $1$ instead) since this interrupts (preempts) the
service of type-$1$ VMs.
There have been two prior approaches to non-preemptive scheduling, namely,
MaxWeight approach \cite{MS2014,MSY2012,Siva13}, and randomized sampling approach
\cite{G2016}. In the rest of the paper, we use the terms VMs and jobs
interchangeably.

\textit{MaxWeight approach.} This approach is based on the classical
MaxWeight scheduling \cite{tassiulas1992stability}.
However unlike scheduling in data packet networks,
here a MaxWeight schedule cannot be used at arbitrary points in time since
it might cause preemption of jobs already in service.
Recent work \cite{MS2014,Siva13} proposes using the MaxWeight
schedule at instances when the servers becomes empty (the so-called refresh times),
however the approach requires using a MaxWeight schedule at times
when \textit{all} the servers become empty simultaneously
(the so-called global refresh times). This requires some
form of synchronization among the servers to set the
MaxWeight schedule at the same time. Further, such global
refresh times become extremely infrequent in large-scale server systems,
thus causing large queues and delays in scheduling.
There is no proof that MaxWeight based on local refresh times
(i.e. when each server chooses a MaxWeight schedule locally at its own refresh time)
is stable in general. In fact, it was suggested in \cite{G2016} that it might be unstable.
Also the approach requires finding the MaxWeight schedule
which in our setting requires solving a Knapsack problem
which is a hard combinatorial problem~\cite{knapsack}.

\textit{Randomized sampling approach.} A randomized sampling approach
was proposed in \cite{G2016} which has low complexity
and can provide high throughput. The idea is that each queue
samples the servers at random and places a token in the server
if it can fit a job in the sampled server. Token acts as place
holder for a future job arrival and reserves resources
for future job of that type for some time duration. When a job arrives, it is placed in a token of that type, if there is any, otherwise it is queued.
The sampling rate used by a queue depends on its size,
i.e, as a queue builds up, it samples the servers faster.
The algorithm is proved to be throughput optimal however in
general it suffers from long convergence time and excessive
queue sizes/delays.

\subsection{Contributions}
The main contributions of this work are summarized below.

\begin{itemize}
\item \textbf{A scalable non-preemptive scheduling algorithm.} We provide a scalable non-preemptive scheduling algorithm that can provide high throughput and low delay. Each server makes its scheduling decisions locally independently of the other servers based on a Knapsack or an approximated Knapsack solution (e.g. a greedy low-complexity solution). The key ingredient of our algorithm is a new construct of refresh times. Specifically each server actively estimates the right moments in time that it needs to reset its schedule and stops scheduling to allow the schedule to be renewed when the server becomes empty.
\item \textbf{Throughput-delay-complexity tradeoff.} We formally prove the fraction
of the throughput region that our algorithm can achieve. Specifically, given an approximation algorithm for solving the Knapsack problem with approximation ratio $r\in (0,1]$, our algorithm can provide $\beta r$ fraction of the maximum throughput where $\beta$ can be tuned to provide tradeoff between throughput and delay. Any general off-the-shelf approximation algorithm for the Knapsack problem can be used as subroutine in our scheduling algorithm, with $\beta \in (0,r)$, however we also present a greedy approximation algorithm for which $\beta \in (0,1)$ works.
\item \textbf{Empirical evaluations.} We provide extensive simulation results, using both synthetic and real traffic traces, that show that our algorithm in fact outperforms prior scheduling algorithms in terms of queuing delay.
\end{itemize}

\subsection{Related Work}
Our work is related to resource allocation in cloud data centers
(e.g.~\cite{cloudpacking1},\cite{VMcompare3},
\cite{chiang,meng2010improving,yazir2010dynamic},
\cite{GSS2016})
and scheduling algorithms in queueing systems
(e.g.~\cite{tassiulas1992stability,bonald2012rateoptimal,ye20101,marsan2002packet,shah2012randomized}). The VM placement in an infinite server system has
been studied in~\cite{S13,SZ13a,SZ13b,ghaderi2014asymptotic}. Four closely related
papers are~\cite{MSY2012},~\cite{Siva13},~\cite{MS2014},~\cite{G2016} where a finite
model of the cloud is studied and preemptive~\cite{MSY2012} and
non-preemptive~\cite{Siva13,MS2014, G2016} scheduling algorithms to stabilize the
system are proposed. The proposed algorithms either rely on the MaxWeight approach and
hence, as explained in Section~\ref{sec-related}, in general suffer from high
complexity and resetting at the global refresh times, or yield excessive queues and
delays in the case of randomized sampling approach. In the case that all the servers
are identical and each server has its own set of queues, it is sufficient to reset the
server configurations at the so-called local refresh times, namely, time instances
when a server becomes empty~\cite{Siva13,MS2014}; however, it is not clear if
operation based on local refresh times is stable in general when the queues are
centralized or the servers are not homogeneous. In fact, operation based on local refresh times can cause instability
(see Example~1 in Simulations, Section~\ref{sec-ineff}).
\subsection{Notations}
In the rest of the paper we use the following notations.
$\left \Vert \cdot \right \Vert$ denotes the Euclidean norm of vectors,
where $\left \Vert \cdot \right \Vert_\infty$ is the $\ell$-infinity norm
which is the maximum element of a vector, and $\left \Vert \cdot \right \Vert_1$
is the $\ell$-1 norm which is the sum of the absolute values of the elements
of the vector. The inner product of two vectors will be denoted by
$\langle \cdot, \cdot \rangle$. $Conv(S)$ is the convex hull of the points
in the set $S$. $|S|$ is the cardinality (the number of elements) of the
set $S$. $\mathbf{0}_n$ is a zero vector of size $n$. $\mathds{1}(E)$ is the
indicator function which is $1$ if condition $E$ is true and 0 otherwise.
We write $f(x) = o(g(x))$ if $\lim_{x \to 0} \frac{f(x)}{g(x)} = 0$

\section{System Model}
\label{sec-model}

\subsubsection*{Cloud Cluster Model}
We consider a collection of $L$ servers denoted by the set $\mathcal{L}$. 
Each server $\ell \in \mathcal{L}$ has a limited capacity for various resource 
types (e.g., memory, CPU, storage, etc.). We assume there are $R$ different 
types of resources. Servers could be inhomogeneous in terms of their capacities.

\subsubsection*{VM-based Job Model} There is a collection of $J$
VM types denoted by the set $\mathcal{J}$. Each VM type $j \in \mathcal{J}$
requires fixed amounts of the various resources.
So each VM type is a $R$-dimensional vector of resource requirements.

\subsubsection*{Job (VM) Arrivals and Service Times}
Henceforth, we use the terms job and VM interchangeably.
We assume VMs of type $j$ arrive according to a Poisson process
with rate $\lambda_j$. The highest rate among them is denoted
by $\lambda_{max} := \max_j{\lambda_j}$.
Each VM must be placed in a server that has enough remaining
resources to accommodate it. Once a VM of type $j$ is placed in server,
it departs after an exponentially distributed amount of time (service time)
with mean $1/\mu_j$, independently of the other existing VMs in the server.
We will also define the maximum mean service time as $T := \max_j{1/\mu_j}$
and the maximum service rate as $\mu_{max} := \max_j{\mu_j}$. The Poisson and
exponential assumptions are for simplicity and we will in fact broaden
the results to more general distributions later in Section~\ref{sec:generalization}.

\subsubsection*{Server Configuration and System Configuration}
We denote by $k_j^\ell$ the number of type-$j$ VMs that are
accommodated by server $\ell$. For each server $\ell$, a vector
$\mathbf{k}^\ell=(k_1^\ell, \cdots, k_J^\ell) \in \mathds{N}_0^J$
is said to be a feasible configuration if the server can simultaneously
accommodate $k_1^\ell$ type-$1$ VMs, $k_2^\ell$ type-$2$ VMs, ...,
$k_J^\ell$ type-$J$ VMs, without violating its capacity.
\change{
	A feasible configuration is said to be \textit{maximal}
	if no further VM can be added to the configuration without violating the server's capacity.
}
We also define the system configuration as a matrix 
$\mathbf{k} \in \mathds{N}_0^{L\times J}$ whose $\ell$-th 
row ($\mathbf{k}^{\ell}$) is the configuration of server $\ell$.

We use $\mathcal{K}_\ell$ to denote the set of all feasible configurations
for server $\ell$ excluding the $0$-configuration $\mathbf{0}_J$, and $\mathcal{\bar K}_\ell$ to denote $\mathcal{K}_\ell \cup \{\mathbf{0}_J\}$.
Note that we do not necessarily need the resource requirements of VMs in a configuration to be additive (vector addition),
we only require the monotonicity of the feasible configurations,
i.e., if $\mathbf{k}^{\ell} \in \mathcal{\bar K}_\ell$,
and $ {\mathbf{k}^\prime}^{\ell} \leq \mathbf{k}^{\ell}$ (component-wise),
then ${\mathbf{k}^\prime}^{\ell} \in \mathcal{\bar K}_\ell$. Clearly 
monotonicity includes additive resource requirements as a special case.

\subsubsection*{Queueing Dynamics and Stability}
When jobs arrive, they are queued and later served by the servers.
We use $Q_j(t)$ to denote the number of type-$j$ jobs waiting in the
queue to get service. The vector of all queue sizes at time $t$ is denoted by
$\mathbf{Q}(t)$.
$Q_j(t)$ follows the usual dynamics
\[
Q_j(t)=Q_j(t_0)+A_j(t_0,t)-D_j(t_0,t),
\]
where $A_j(t_0,t)$ is the number of type-$j$ jobs arrived
from time $t_0$ up to time $t$ and $D_j(t_0,t)$
is the number of type-$j$ jobs departed from queue in the same time interval.
\change{
The system is said to be stable if the queues remain bounded in the sense that
\be \label{eq:mean sense}
\lim_{t \to \infty} \sup \mathds{E}\left[\sum_j Q_j(t)\right] < \infty.
\ee
}

\change{
A vector of arriving rates $\bm{\lambda}$ and a vector of mean service
times $1/\bm{\mu}$ is said to be supportable
if there exists a scheduling algorithm under which the system is stable.
}
Let $\rho_j=\lambda_j/\mu_j$ be the workload of type-$j$ jobs.
We will define the capacity (throughout) region
of the cluster as
\begin{equation}
\label{eq-C}
\mathcal{C} = \{ \mathbf{x} \in \mathbb{R}^J_+ : \mathbf{x} =
\sum_{\ell \in \mathcal{L}} \mathbf{x}^\ell,\ \mathbf{x}^\ell \in
Conv(\mathcal{\bar K}^\ell), \ell \in \mathcal{L}
   \},
\end{equation}
where $Conv(\cdot)$ is the convex hull operator. It has been shown
~\cite{Siva13,MS2014,MSY2012} that the set of supportable workloads
$\bm{\rho}=(\rho_1, \cdots \rho_J)$ is the interior of $\mathcal{C}$.
We also define $\mathcal{C}_\beta$ as the $\beta$ fraction of the capacity region,
i.e., $\mathcal{C}_\beta = \beta \mathcal{C}$, for $0 < \beta \leq 1$.

\section{Basic Algorithm and Main Result} \label{alg}
In this section, we present our non-preemptive scheduling algorithm and state the main result regarding its performance. Before describing the algorithm, we make two definitions.

\vspace{0.08 in}
\textsc{Definition 1} (weight of a configuration). The weight of configuration
$\mathbf{k}^\ell$ for server $\ell$, given a queue size vector $\mathbf{Q}$,
is defined as
\begin{equation} \label{eq-f}
f(\mathbf{k}^\ell, \mathbf{Q}) := \sum_{j\in \calJ} Q_j k_j^\ell.
\end{equation}

\change{
\textsc{Definition 2} ($r$-max weight configuration). Given a constant
$r \in (0,1]$, and a queue size vector $\mathbf{Q}$, an
$r$-max weight configuration for server $\ell$ is a feasible configuration
$\mathbf{k}^{(r) \ell}\in \mathcal{K}_\ell$ such that
\begin{equation}\label{eq-maxconf}
f(\mathbf{k}^{(r) \ell}, \mathbf{Q}) \geq r f(\mathbf{k}^\ell, \mathbf{Q})
,\ \forall \mathbf{k}^\ell \in \mathcal{K}_\ell.
\end{equation}
}
Note that by Definition 2, an $r$-max weight configuration, is also an $r^\prime$-max weight configuration, for any $0\leq r^\prime \leq r$.

Various approximation algorithms exist that can provide an $r$-max weight
configuration. In Section~\ref{sec-optconf}, we will elaborate further and describe several low complexity approaches to solve (\ref{eq-maxconf}),
but for now assume that such an approximation algorithm exists and is used
as a subroutine in our scheduling algorithm in a black box fashion.

Under our scheduling algorithm, each server at any time is either in an \textit{active} period or in a \textit{stalled} period, defined below.
We will also refer to the state of a server as active or stalled depending on the period in which it is at a certain time.

\textbf{Active period:} In an active period, the server schedules jobs
from the queues according to a \textit{fixed} configuration.
Formally, let the configuration of server $\ell$ in an active period be
$\tilde{\mathbf{k}}^{\ell}=(\tilde{k}^{\ell}_j: j \in \mathcal{J})$.
The server can contain at most $\tilde{k}^{\ell}_j$ jobs of type $j$,
$j\in \mathcal{J}$, at any time.
If there are not enough type-$j$ jobs in the system, the server
reserves the remaining \textit{empty slots} for future type-$j$ arrivals.
We use $\bar{\mathbf{k}}^\ell(t)=(\bar{k}^{\ell}_j(t); j\in \calJ)$
to denote the actual number of jobs in the server $\ell$ at time $t$.
By definition, $\bar{\mathbf{k}}^\ell(t) \leq \tilde{\mathbf{k}}^{\ell}$
(component-wise) at any time $t$ during the active period of server $\ell$.

\textbf{Stalled period:} In a stalled period, the server
does not schedule any more jobs, even if there are jobs waiting
for service that can fit in the server,
and it only processes jobs which already exist in the server.
The stalled period of the server ends when all the existing jobs
in the server finish their service and leave,
at which point the server will enter a new active period.

Note that by the above definitions, an arriving job of type $j$ will not be queued
(i.e., it enters the queue but immediately gets service) if there is an \textit{empty
slot} available for it in any of the active servers (i.e., if there is a server $\ell$
such that $\tilde{k}^\ell_j-\bar{k}^\ell_j(t)\geq 1$), as it will be scheduled in one
of the empty slots immediately.
Also the change of configuration in a server can only happen when the server
is empty and stalled and that change results in a transition from a stalled period to
an active period.
We will refer to these transition times as \textit{configuration reset times}.

Our scheduling algorithm determines:
(1) the time at which a server must go from active to stalled,
(2) the time at which a server must go from stalled to active,
and (3) the server configuration used during the active period
when the server goes from stalled to active.


\begin{enumerate}[leftmargin=*]
\item \textbf{Transition from active to stalled}.
Suppose server $\ell$ is in an active period with configuration
$\tilde{\mathbf{k}}^\ell$. The server makes a transition to a stalled period
if upon departure of a job from the server at time $t$,
\begin{equation} \label{st_cond}
f(\tilde{\mathbf{k}}^\ell, \mathbf{Q}(t)) <
\beta f(\mathbf{k}^{(r) \ell}(t), \mathbf{Q}(t)),
\end{equation}
where $\mathbf{k}^{(r) \ell}(t)$ is an $r$-max configuration given the queue
size vector $\mathbf{Q}(t)$ (based on Definition 2), and  $0 < \beta <1$ is a constant
which is a parameter of the algorithm.
In other words, transition occurs when the weight of the
active server's configuration $\tilde{\mathbf{k}}^\ell$
becomes worse than $\beta$ fraction of the weight of the
$r$-max weight configuration $\mathbf{k}^{(r) \ell}(t)$
computed at the time of job departure $t$.
\textit{Note that condition~(\ref{st_cond}) is only checked when a job
hosted in server $\ell$ is completed.}


\item \textbf{Transition from stalled to active.}
Suppose a server is in a stalled period. When the server becomes empty (i.e., its existing jobs finish service), the server makes a transition to an active period.


\item \textbf{Server configuration during an active period}.
Suppose server $\ell$ enters an active period at time $t_{(a)}$.
The configuration of server $\ell$ for the entire duration of its active period, $\tilde{\mathbf{k}}^\ell$, is \textit{fixed} and
set to $\mathbf{k}^{(r) \ell}(t_{(a)})$, an $r$-max weight configuration based on the
queues at time $t_{(a)}$. Note that in Definition~2, the zero configuration 
$\mathbf{k}^\ell = \mathbf{0}_J$ is not selected, even when all the queues are empty.
\end{enumerate}
Algorithm~\ref{alg1} gives a description of our algorithm.

\begin{algorithm}
\caption{Basic Non-preemptive Scheduling}
\label{alg1}
\begin{flushleft}
	\textit{When a job of type $j$ arrives at time $t$:}
\end{flushleft}

\begin{algorithmic}[1]
\STATE Add the job to the queue $j$
\IF {exists \textit{empty slots} for type-$j$ jobs}
	\STATE Schedule the job in the first empty slot.
\ENDIF
\end{algorithmic}

\begin{flushleft}
	\textit{When a job of type $j$ in server $\ell$ is completed at time $t$:}
\end{flushleft}

\begin{algorithmic}[1]
\IF{$\ell$ is \textit{active} with configuration $\tilde{\mathbf{k}}^\ell$}
     \IF{condition~(\ref{st_cond}) holds}
     	\STATE Switch $\ell$ to \textit{stalled}.
	 \ELSE
	 	\STATE Schedule a type-$j$ job in server $\ell$ from queue $j$. If queue $j$ is empty, register an \textit{empty slot} of type $j$ in server $\ell$.
     \ENDIF
\ENDIF
\IF{$\ell$ is empty and \textit{stalled}}
	\STATE Switch $\ell$ to \textit{active}.
	\STATE Find an $r$-max weight configuration $\mathbf{k}^{(r) \ell}$.
	\STATE Set the configuration of server $\ell$ during its active
	period to be fixed and equal to $\mathbf{k}^{(r) \ell}$.
	\FOR{$j \in \mathcal{J}$}
		\STATE{Schedule ${k_j}^{(r) \ell}$ jobs of type $j$ in server $\ell$.}
If there are not enough jobs in queue $j$, register an \textit{empty slot} for each unused slot.
	\ENDFOR
\ENDIF
\end{algorithmic}
\end{algorithm}
\vspace{0.07 in}
\textsc{Remark 1} (\textit{choice of $r$ and $\beta$}):
The parameter $r$ provides a flexibility in solving the optimization
(\ref{eq-maxconf}) depending on the server and job profiles. In general,
it might be difficult to find the max weight configuration for $r=1$
in~(\ref{eq-maxconf}) (this is the so-called Knapsack problem~\cite{knapsack}),
but there are greedy algorithms that can guarantee that the configuration 
will be $r$-max weight for some $r<1$ (see Section~\ref{sec-optconf}).

The parameter $\beta$ that appears in condition~(\ref{st_cond})
controls how often servers
transit to stall period and as we will prove later controls
what fraction of the maximum throughput (capacity) region is
achievable. Higher $\beta$ makes a server stall more often,
which increases the overall delay
of jobs waiting to get service, however it can achieve higher throughput. Therefore $\beta$ can be tuned to provide a tradeoff between throughput and average delay.

\vspace{0.07 in}
\textsc{Remark 2} (\textit{configuration reset times}): The prior approach~\cite{MS2014} is based
on finding the max weight configuration (corresponding to $r=1$ in~(\ref{eq-maxconf})),
and changing the configuration of a server at the so-called refresh times
when the servers become empty. However their proof of stability requires
resetting the server configuration at `global' refresh times when \textit{all}
the servers become empty at the same time. Such times could be extremely
rare when the system size is large. Resetting the server configurations
at their local refresh times (i.e., when each server itself is empty)
cannot guarantee stability, in fact we can give examples that show that it
becomes unstable (see Example~1 in Section~\ref{sec-ineff}). Algorithm~\ref{alg1}
does not require synchronization among the reset times of servers and every
server can reset its configuration locally based on its local state information.
\change{Intuitively our method works because each server actively estimates the right moment in time that it needs to reset its configuration, and stops scheduling to allow the configuration to reset, something that doesn't happen in the other methods.
}

The following theorem states the main result about the performance of the algorithm.
\begin{theorem}\label{theorem}
Consider Algorithm \ref{alg1} with parameter $r \in (0,1]$ and $0 < \beta < r$.
Then the algorithm can support any workload vector $\bm{\rho}$
in the interior of $\mathcal{C}_{r\beta}$ ($r\beta$-fraction of the capacity region $\mathcal{C}$).
\end{theorem}

\section{Proof of Main Result}\label{sec-lyap}
\change{The proof of Theorem~\ref{theorem} is based on Lyapunov analysis. The idea is to show that for large enough queue sizes,
the servers will be in active periods most of the time and their 
negative contribution to the drift of Lyapunov function will outweigh the positive 
contribution of stalled periods. The challenge is that servers, under 
Algorithm~\ref{alg1}, make their (active, stalled) decisions locally without 
coordination. Despite this, we are still able to show that all the servers will be 
active simultaneously for sufficiently large fraction of time. The proof follows 
3 main steps as follows.}
\subsection{System state}
The system state at any time is given by
\begin{equation}\label{eq-state}
\mathbf{S}(t) = \left(\mathbf{Q}(t), \bar{\mathbf{k}}(t),
	\tilde{\mathbf{k}}(t), \mathbf{I}(t) \right),
\end{equation}
where $\mathbf{Q}(t)$ is the vector of queue sizes (i.e., jobs waiting
to get service), $\bar{\mathbf{k}}(t)$ denotes the existing jobs in
the servers, $\tilde{\mathbf{k}}(t)$ is the system configuration, and
$\mathbf{I}(t)$ indicates which server is active or stalled, i.e.,
$I_\ell(t)=1$ if server $\ell$ is in active period, and is zero if
it is stalled. Under Algorithm~\ref{alg1}, the process $\mathbf{S}(t)$
evolves as a continuous-time and irreducible Markov chain. Note that
when $I_{\ell}(t)=1$, if $\bar{k}^{\ell}_j(t) < \tilde{k}^{\ell}_j(t)$
for some type $j$ in server $\ell$ (i.e., there is at least one
\textit{empty slot} for type-$j$ VMs), that necessarily implies that
$Q_j(t) = 0$. For notational compactness, throughout the proofs, we use
$\mathds{E}_{\mathbf{S}(t)}$ to denote the conditional expectation,
given state $\mathbf{S}(t)$.
\subsection{Duration of overlapping active periods among servers}
\change{We show that as queues get large, the accumulated duration of overlapping 
	active periods (i.e, durations when \textit{all} servers are active simultaneously) 
	will become longer while the accumulated duration of stalled periods 
	remains bounded, with high probability.
	To show this, we analyze the active/stalled periods over an interval of
	length $NT$, where $T=\max_j{1/\mu_j}$ and $N$ is a large constant to be determined.}

The following Lemma is essential to our proof.
\change{
\begin{lemma}
\label{lem_active}
Suppose server $\ell$ becomes active at time $t_{(a)}$. There exists
a constant $C>0$ such that the server will remain active during the interval
$[t_{(a)}, t)$ if
\begin{equation}
\left \Vert \mathbf{A}(t_{(a)}, t) \right \Vert_\infty +
\left \Vert \mathbf{D}(t_{(a)}, t) \right \Vert_\infty <
C \left \Vert \mathbf{Q}(t_{(a)}) \right \Vert,
\nonumber
\end{equation}
where $\mathbf{A}(t_{(a)}, t) $ and $\mathbf{D}(t_{(a)}, t) $ are respectively the vector of number arrivals
and departure during $[t_{(a)},t)$.
\end{lemma}
}
\begin{proof}

In this proof, we use the inner-product notation to represent the function
$f$ defined in (\ref{eq-f}), i.e
$f(\mathbf{k}^\ell, \mathbf{Q}(t)) = \langle \mathbf{k}^\ell, \mathbf{Q}(t) \rangle$,
to make the vector interpretation easier.

\change{
At time $t_{(a)}$ when server becomes active, its configuration is set to
$\tilde{\mathbf{k}}^{\ell}(t_{(a)})$ which by Definition~2 satisfies
\begin{equation} \label{ineq1}
\langle \tilde{\mathbf{k}}^{\ell}(t_{(a)}) - r\mathbf{k}^\ell, \mathbf{Q}(t_{(a)})
 \rangle \ge 0;\ \forall \mathbf{k}^\ell \in \mathcal{K}^\ell.
\end{equation}
}
For the server to become stalled for the first time at job departure time
$t_{(s)}>t_{(a)}$, the condition (\ref{st_cond}) should hold for the first
time at departure time $t_{(s)}$. This implies that at time $t_{(s)}$,
\change{
\begin{equation} \label{ineq2}
\exists \mathbf{k^\ell} \in \mathcal{K}^{\ell}: \langle
\tilde{\mathbf{k}}^{\ell}(t_{(a)}) - \beta \mathbf{k}^\ell, \mathbf{Q}(t_{(s)})
\rangle < 0,
\end{equation}
which is clearly satisfied by at least the choice of
$\mathbf{k^\ell}=\mathbf{k}^{(r) \ell}(t_{(s)})$
($r$-max weight configuration at time $t_{(s)}$).
}
Hence, as a \textit{sufficient} condition, the server will certainly never get stalled (it remains active) during $[t_{(a)},t_{(s)})$ if at any time $t \in [t_{(a)},t_{(s)})$
\begin{equation} \label{ineq3}
\forall \mathbf{k}^\ell \in \mathcal{K}^\ell:\ \langle
\tilde{\mathbf{k}}^{\ell}(t_{(a)}) - \beta \mathbf{k}^\ell, \mathbf{Q}(t) \rangle \ge 0.
\end{equation}

Figure \ref{fig-lem} gives a visualization of the boundaries
of the Inequalities (\ref{ineq1}) and (\ref{ineq2}), in two dimensions. 
One can see that if $\beta=r$ the boundaries
will be identical, while as $\beta$ becomes less than $r$,
and approaches $0$, the gap between the boundaries becomes
wider, and server $\ell$ stalls less frequently.
Given a fixed $\mathbf{k}^\ell$, the boundaries are hyperplanes
with respect to variable $\mathbf{Q}$ and the angle between them,
as highlighted in Figure \ref{fig-lem}, is
\begin{equation} \label{eq-theta}
\theta_{\mathbf{k}^\ell} = \arccos{\frac{\langle
		\tilde{\mathbf{k}}^{\ell}(t_{(a)}) - r\mathbf{k}^\ell,
  \tilde{\mathbf{k}}^{\ell}(t_{(a)}) - \beta \mathbf{k}^\ell \rangle}
{\left \Vert \tilde{\mathbf{k}}^{\ell}(t_{(a)}) - r\mathbf{k}^\ell \right \Vert \left
  \Vert \tilde{\mathbf{k}}^{\ell}(t_{(a)}) - \beta \mathbf{k}^\ell  \right \Vert}} > 0.
\end{equation}
\change{
This implies that the server will certainly remain active during
$[t_{(a)}, t)$ as long as the change in the queue size vector $\mathbf{Q}(t_{(a)})$,
due to arrivals and departures during $[t_{(a)}, t)$, does not move it from
the green region to the red region, a distance of length $L$ as highlighted
in Figure~\ref{fig-lem}. Since distance $L$ is at least
$\sin(\theta_{\mathbf{k}^\ell}) \left \Vert \mathbf{Q}(t_{(a)})\right \Vert$,
the server is guaranteed to remain active, if the change in the norm of the
queue size vector is less than this quantity. 
This should be true for every possible choice of
$\mathbf{k}^\ell$, i.e.,
$\left \Vert \mathbf{Q}(t_{(a)}) - \mathbf{Q}(t) \right \Vert <
\sin\left(\min_{\mathbf{k}^\ell \neq \tilde{\mathbf{k}}^{\ell}(t_{(a)}) } \theta_{\mathbf{k}^\ell} \right)
\left \Vert \mathbf{Q}(t_{(a)})\right \Vert$,
or equivalently
}
\be \label{eq:ca}
\left \Vert \mathbf{A} - \mathbf{D} \right \Vert <
C_a \left \Vert \mathbf{Q}(t_{(a)}) \right \Vert,
\ee
where $C_a = \sin\left(\min_{\mathbf{k}^\ell \in \mathcal{K}_\ell,
	\mathbf{k}^\ell \neq \tilde{\mathbf{k}}^{\ell}(t_{(a)}) }
	\theta_{\mathbf{k}^\ell} \right).$
Note that $C_a$ is a \textit{strictly positive constant}, because $r > \beta > 0$
and $\mathbf{k}^\ell  \nparallel  \tilde{\mathbf{k}}^{\ell}(t_{(a)})$
($\nparallel$ means not parallel).
The case $\mathbf{k}^\ell  \parallel  \tilde{\mathbf{k}}^{\ell}(t_{(a)})$
never happens. To arrive at a contradiction, suppose $\mathbf{k}^\ell  \parallel  \tilde{\mathbf{k}}^{\ell}(t_{(a)})$, which implies
$\tilde{\mathbf{k}}^{\ell}(t_{(a)}) = C_k \mathbf{k}^\ell $ for some
constant $C_k$. On the other hand by (\ref{ineq1}),
$\langle\tilde{\mathbf{k}}^{\ell}(t_{(a)}), \mathbf{Q}(t_{(a)}) \rangle \ge
r \langle\mathbf{k}^\ell , \mathbf{Q}(t_{(a)}) \rangle$.
Therefore it holds that $C_k\ge r > \beta$ and
\begin{equation}
\begin{aligned}
&\langle\tilde{\mathbf{k}}^{\ell}(t_{(a)}), \mathbf{Q}(t) \rangle =
C_k \langle\mathbf{k}^\ell, \mathbf{Q}(t) \rangle \ge \beta \langle\mathbf{k}^\ell,
\mathbf{Q}(t) \rangle,
\end{aligned} \nonumber
\end{equation}
which implies $\langle\tilde{\mathbf{k}}^{\ell}(t_{(a)}) - \beta \mathbf{k}^\ell, \mathbf{Q}(t) \rangle \ge 0$,
so inequality (\ref{ineq2}) is never true and configuration can never change to $\mathbf{k}^\ell$.

Note that
$\left \Vert \mathbf{A} - \mathbf{D} \right \Vert \le
\left \Vert \mathbf{A} \right \Vert + \left \Vert \mathbf{D} \right \Vert \le
\sqrt{J}(\left \Vert \mathbf{A} \right \Vert_\infty + \left \Vert \mathbf{D} \right \Vert_\infty)$. Thus
a stricter condition than (\ref{eq:ca}) that ensures the server remains active during $[t_{(a)},t)$ is the one given by the statement of Lemma
by choosing $C = \frac{C_a}{\sqrt{J}}$.

\end{proof}

\begin{figure}[t]
	\centering
	\begin{minipage}[t]{0.41\textwidth}
		\centering
		\includegraphics[width=0.92\columnwidth]{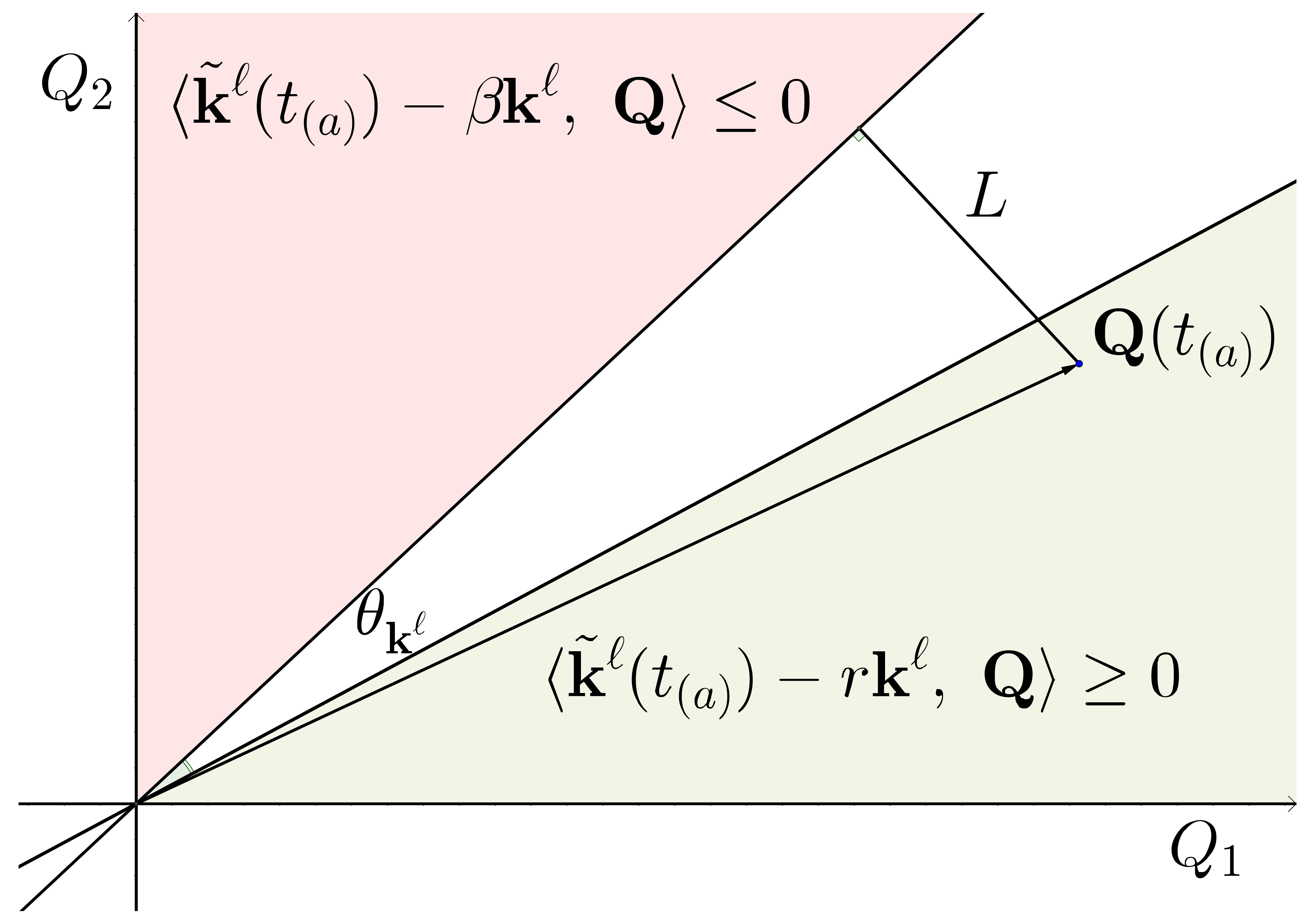}
		\caption{
			\change{
				Illustration of proof of Lemma~\ref{lem_active}	for $2$ dimensions.
				When server becomes active, queue size vector $\mathbf{Q}(t_{(a)})$
				is in the green region. Server will stall if the queue size vector
				reaches the red region for a configuration $\mathbf{k}^\ell$.
			}
		}
		\label{fig-lem}
	\end{minipage}\hfill \vspace{0pt}
	\begin{minipage}[t]{0.56\textwidth}
		\centering
		\includegraphics[width=1.0\columnwidth]{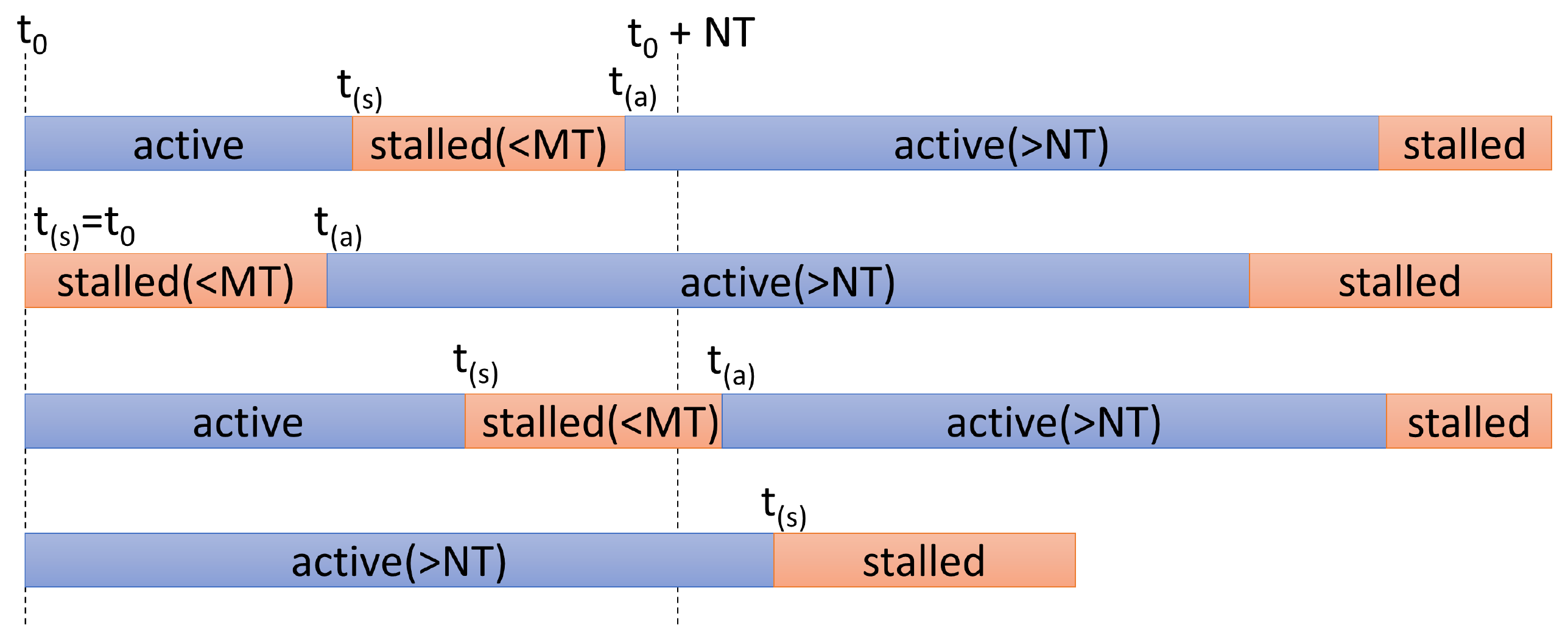}
		\caption{A subset of event $E_{\mathbf{S}(t_0),M,N}$. Any server stalls for `at most' $MT$ amount 
			of time and is active for `at least' $NT$ amount of time afterwards. 
			All possible cases are illustrated above. 
			$t_{(s)}$ ($\geq t_0$) is the entrance time 
			to a stalled period, and $t_{(a)}$ is the entrance time to the subsequent active period).}
		\label{fig-emn}
	\end{minipage}
\end{figure}

Next, we bound the duration of time that servers are active simultaneously during
an interval $[t_0, t_0+NT]$. Define $E_{\mathbf{S}(t_0),M,N}$ as the event
that in this time interval, every server will be stalled at most once
and for at most $MT$ time duration, for some positive constant $M$, 
given the initial state $\mathbf{S}(t_0)$. Note that this will imply 
that the total accumulative amount of time that at least one server 
is stalled in the time interval is less than $LMT$. We show that 
$E_{\mathbf{S}(t_0),M,N}$ is almost certain for large enough values 
of $M$ and $\Vert\mathbf{Q}(t_0)\Vert$.
\begin{proposition}\label{prop:Mq}
Given any $\epsilon \in (0,1)$, there are constants $C_1$ and $C_2$ such that
$\mathds{P}(E_{\mathbf{S}(t_0),M,N}) > 1 - \epsilon$, if
\begin{equation} \label{Mq-cond-def}
\begin{aligned}
M > - \log \left(\epsilon\right) + C_1; \ \left \Vert \mathbf{Q}(t_0) \right \Vert
  > \frac{N}{\epsilon} C_2.
\end{aligned}
\end{equation}
\end{proposition}

\begin{proof}
	 A sketch of the proof is as follows:
	\begin{enumerate}
	\item The number of jobs in any server is bounded and
	their expected time of service is also bounded, so once
	a server enters a stalled period, it will almost
	certainly enter an active period again in finite time.
	\item Using Lemma~\ref{lem_active}, we can argue that the minimum
	expected length of an active period
	is proportional to the length of queue size vector at the beginning
	of the active period.
	\item To bound the probability of event $E_{\mathbf{S}(t_0),M,N}$, it
	suffices to consider its following subevent: if a server
	becomes stalled at a time in the interval $[t_0, t_0 + NT]$,
	it becomes empty within $MT$ amount of time,
	and once the server becomes active, it remains active for at least
	$NT$ amount of time. This event is a subset of $E_{\mathbf{S}(t_0),M,N}$, 
	as illustrated in Figure~\ref{fig-emn},
	which considers all possible transition times between active
	and stalled periods in the time interval $[t_0, t_0+NT]$.
	\end{enumerate}

The rest of the proof follows from basic probability calculations. The detailed proof can be found in Appendix~\ref{proof:Mq}.
\end{proof}


\subsection{Lyapunov analysis}
To prove the stability of the algorithm, we will use the following Lyapunov function
\begin{equation}\label{lyap}
V(t) = \sum_j \frac{{Q}_j(t)^2}{2\mu_j}.
\end{equation}
Define the infinitesimal generator~\cite{SDE} of the Lyapunov fucntion $V(t)$ as
\begin{equation} \label{eq-infA}
A V(t) := \lim_{u \to 0} \frac{\mathds{E}_{\mathbf{S}(t)}[V(t+u)]-V(t)}{u}
\end{equation}
Then we show the following lemma.
\begin{lemma} \label{lem-inf}
At any time $t$,
\begin{equation} \label{eq-infB}
A V(t) \le  \sum_j [{Q}_j(t)( \rho_j - \sum_\ell I_\ell(t) \tilde{k}_j^\ell(t))]+B_2,
\end{equation}
for a positive constant $B_2$. Recall that $I_\ell(t)$ is the indicator function
defined in the system state (\ref{eq-state}).
\end{lemma}

\begin{proof}
See Appendix~\ref{proof:lem-inf} for the proof.	
\end{proof}

In Algorithm \ref{alg1}, transition from active to stalled could happen only at the departure times of the jobs hosted in the server. Nevertheless, the weight of
the server configuration at any time in the active period, is still `roughly'
at least $\beta r$ fraction of the max weight configuration.
The following lemma formalizes this statement.
\change{
\begin{lemma}\label{prop:weight_const}
	Suppose server $\ell$ is active and has configuration
	$\tilde{\mathbf{k}}^\ell $ for the duration of its
	active period. Let $E_{B_1, \ell}$ be the event that
	$f(\tilde{\mathbf{k}}^\ell, \mathbf{Q}(t)) >
	\beta r f(\mathbf{k}^\ell, \mathbf{Q}(t)) - B_1$,
	for any $\mathbf{k}^\ell \in \mathcal{K}^\ell$ and at
	any time $t$ in the active period. Then given any
	$\epsilon \in (0,1)$, there exist constants $C_3, C_4>0$
	such that $\mathds{P}(E_{B_1, \ell}) > 1-\epsilon$
	if $B_1 > -C_3\log{\epsilon} + C_4$.
\end{lemma}
}
\begin{proof}
See Appendix~\ref{proof:weight_const} for the proof.	
\end{proof}

Equipped with the Lemmas and Propositions above,
we analyze the drift of the Lyapunov function
in the following proposition.
\change{
\begin{proposition}\label{prop:Nq}
Consider the Lyapunov function $V(t)$ defined in (\ref{lyap}).
Given the workload $\bm{\rho}$ inside the $r\beta$ fraction of the
capacity region $\mathcal{C}$, $t_f=t_0+NT$, and any $\delta>0$,
\[\mathds{E}_{\mathbf{S}(t_0)} \left[V(t_f) - V(t_0)\right] < -\delta \]
if
\begin{equation} \label{Nq-cond-def}
\begin{aligned}
N > M C_5,\ \left \Vert \mathbf{Q}(t_0) \right \Vert > C_6(M, N, \delta),
\end{aligned}
\end{equation}
where $C_5$ is a constant and $C_6$ is a function of $M$, $N$, $\delta$.
\end{proposition}
}

\begin{proof}

Let the initial system state be $\mathbf{S}(t_0)$ with initial
queue size vector $\mathbf{q}_0$ and $t_f=t_0+NT$. Then by application
of Dynkin's Theorem \cite{SDE}, applied to Lemma~\ref{lem-inf},

\change{
\begin{equation}\label{eq:drift}
\begin{aligned}
&\mathds{E}_{\mathbf{S}(t_0)} \left[V(t_f) - V(t_0) \right]
= \mathds{E}_{\mathbf{S}(t_0)}\left[\int_{t=t_0}^{t_f} AV(t) dt\right] \le \\
&\mathds{E}_{\mathbf{S}(t_0)}\left[\int_{t=t_0}^{t_f} \left( \sum_j Q_j(t)\rho_j -
\sum_\ell I_\ell(t) \sum_j Q_j(t) \tilde{k}_j^\ell(t) \right) + B_2 dt \right].
\end{aligned}
\end{equation}
}
Given a workload $\bm{\rho}$ inside the $r\beta$ fraction
of the capacity region, there exists an $\epsilon$
such that $\bm{\rho} < (1-\epsilon)r\beta \sum_\ell \mathbf{x}^\ell$
for $\mathbf{x}^\ell$ in $conv(\mathcal{K}_\ell)$.
We denote by $E_{(a)}(t)$ the event that \textit{all} servers are active
at time $t$, by $E_{(s)}(t)$ the events that \textit{at least one} is stalled
and by $\mathbf{k}^{\star \ell}(t) = ({k^\star}_1^{\ell}, \cdots, {k^\star}_J^{\ell})$
a max weight configuration at time $t$, i.e
$f(\mathbf{k}^{\star \ell}(t), \mathbf{Q}(t))
\geq f(\mathbf{k}^\ell, \mathbf{Q}(t)),\ \forall \mathbf{k}^\ell \in \mathcal{K}_\ell$.
Note that by definition, $\mathbf{k}^{\star \ell}(t)$ is an $r$-max weight
configuration for $r=1$. Recall the definition of event $E_{B_1, \ell}$ in Lemma~\ref{prop:weight_const}. With a minor abuse of notation, we use $E_{(i) B_1, \ell}$ to denote $E_{B_1, \ell}$ in the $i$-th active period during the interval $(t_0, t_f)$, $i=1, 2, \cdots$.  
Then we can bound the second term of the expectation above as
\begin{equation} \label{eq-integralk}
\begin{aligned}
&\mathds{E}_{\mathbf{S}(t_0)}\left[\int_{t=t_0}^{t_f} \sum_\ell I_\ell(t)
\sum_j Q_j(t) \tilde{k}_j^\ell(t) dt \right] \ge^{(a)}
\mathds{E}_{\mathbf{S}(t_0)}\left[\int_{t=t_0}^{t_f} \mathds{1}(E_{(a)}(t)) \sum_\ell \sum_j Q_j(t) \tilde{k}_j^\ell(t) dt\right] \ge^{(b)} \\
&\mathds{P}(E_{\mathbf{S}(t_0),M,N}) 
\mathds{E}_{\mathbf{S}(t_0)}\left[\int_{t=t_0}^{t_f} \mathds{1}(E_{(a)}(t))
\sum_\ell \sum_j Q_j(t) \tilde{k}_j^\ell(t) dt
|E_{\mathbf{S}(t_0),M,N}\right] \ge^{(c)} \\
&(1-\epsilon)\mathds{E}_{\mathbf{S}(t_0)}\left[\int_{t=t_0}^{t_f}
\mathds{1}(E_{(a)}(t))
\sum_\ell \mathds{P}(E_{(1) B_1, \ell} | E_{\mathbf{S}(t_0),M,N})
\mathds{P}(E_{(2) B_1, \ell} | E_{\mathbf{S}(t_0),M,N}, E_{(1) B_1, \ell}) \right.\\
&\left.\left(-B_1 +\sum_j Q_j(t) r\beta {k^\star}_j^{\ell}(t) \right) dt
| E_{\mathbf{S}(t_0),M,N} \right] \ge^{(d)} \\
&(1-\epsilon)\mathds{E}_{\mathbf{S}(t_0)}\left[\int_{t=t_0}^{t_f}
 (1-2\epsilon)(1-3\epsilon) \mathds{1}(E_{(a)}(t)) \left(-L B_1 +\sum_\ell \sum_j
Q_j(t) r\beta x_j^\ell \right) dt
| E_{\mathbf{S}(t_0),M,N} \right].
\end{aligned}
\end{equation}
In the above, Inequality (a) holds because we ignore the sum of positive terms
when some of the servers are in active period.
Inequality (b) follows from conditioning on the event $E_{\mathbf{S}(t_0),M,N}$.
In Inequality (c), we have used the fact that
$\mathds{P}(E_{\mathbf{S}(t_0),M,N}) > 1 - \epsilon$
under Lemma~\ref{prop:Mq}, and also the result of Lemma~\ref{prop:weight_const}
with $\mathbf{k}^\ell$ replaced by the max weight configuration
$\mathbf{k}^{\star \ell}(t)$ at time $t$.
Notice that conditioned on the occurrence of event $E_{\mathbf{S}(t_0),M,N}$,
every server could be at most in two active periods in the interval $[t_0, t_0 +NT]$, hence we only need to consider events $E_{(1) B_1, \ell}$ and $E_{(2) B_1, \ell}$. 
Finally Inequality (d) uses that
$\mathds{P}(E_{(1)B_1, \ell}|E_{\mathbf{S}(t_0),M,N}) > (1-2\epsilon)$,
which can be inferred from the law of total probability and the fact that
$\mathds{P}(E_{B_1, \ell}) > 1 - \epsilon$ (Lemma~\ref{prop:weight_const}) and
$\mathds{P}(E_{\mathbf{S}(t_0),M,N}) > 1 - \epsilon$ (Proposition~\ref{prop:Mq}). Similarly,  
$\mathds{P}(E_{(2) B_1, \ell} | E_{\mathbf{S}(t_0),M,N}, E_{(1) B_1, \ell}) >
1 - 3\epsilon$.
Thus using (\ref{eq:drift}) and (\ref{eq-integralk}),
the drift can be bounded as follows
\begin{equation}\label{eq:drift2}
\begin{aligned}
&\mathds{E}_{\mathbf{S}(t_0)} \left[V(t_0) - V(t_f) \right] \leq \\
&\mathds{E}_{\mathbf{S}(t_0)} \left[\int_{t=t_0}^{t_f}
\mathds{1}(E_{(a)}(t)) \sum_j Q_j(t)
\left(\rho_j - (1-\epsilon)(1-2\epsilon)(1-3\epsilon) r\beta  \sum_\ell x_j^\ell \right)
\: dt | E_{\mathbf{S}(t_0),M,N} \right] \\
& + \mathds{E}_{\mathbf{S}(t_0)}\left[ \int_{t=t_0}^{t_f} \mathds{1}(E_{(s)}(t))
\sum_j Q_j(t)\rho_j dt | E_{\mathbf{S}(t_0),M,N} \right] + (LB_1 + B_2) NT \le \\
&(N-LM)T \mathds{E}_{\mathbf{S}(t_0)}\left[\max_{t_0\leq t \leq t_f}
{\sum_j Q_j(t)\left(\rho_j - (1-\epsilon)(1-2\epsilon)(1-3\epsilon) r \beta  
	\sum_\ell x_j^\ell \right)}\right] \\
& + LM T \mathds{E}_{\mathbf{S}(t_0)} \left[\max_{t_0\leq t \leq t_f}
{\sum_j Q_j(t)\rho_j}\right]
+ (LB_1 + B_2) NT,
\end{aligned}
\end{equation}
where in the the first inequality, we have used the fact
that events $E_{(a)}(t)$ and $E_{(s)}(t)$
are complementary. As a result we break the integral into
two depending on whether any of the servers is stalled.
In the case that $E_{(s)}(t)=1$, we ignore the departure rates completely.
The last inequality is immediate by noting that by Lemma~\ref{prop:Mq}, the accumulative time duration that $E_{(s)}(t)=1$ is not greater than $MLT$.

Let $v_j = \rho_j - (1-\epsilon)(1-2\epsilon)(1-3\epsilon)\beta \sum_\ell x_j^\ell $,
and vector $\mathbf{v} = (v_1, \cdots, v_J)$. Note that
$\mathbf{v}$ has negative entries for $\epsilon$ small enough (since $\bm{\rho}$ was inside the capacity region),
and $\bm{\rho}$ has positive entries, thus the RHS (Right-Hand-Side) of (\ref{eq:drift2}) is bounded as follows
\begin{equation}
\begin{aligned}
&\mbox{RHS (\ref{eq:drift2})}\leq
(N-LM)T \left( \sum_j (Q_j(t_0) - L K_{max}T \mu_j) v_j \right)
+ LM T\left( \sum_j (Q_j(t_0) + NT \lambda_j) \rho_j \right) + (LB_1 + B_2) NT.
\end{aligned} \nonumber
\end{equation}
Therefore the Lyapunov drift is bounded as
\begin{equation}
\label{drift-cond}
\mathds{E}_{\mathbf{S}(t_0)} [V(t_0) - V(t_f)]
	\le \sum_j C_j(M, N) Q_j(t_0) + C_g(M,N),
\end{equation}
where
\begin{equation} \label{eq-Cf}
\begin{aligned}
&C_j(M, N) = (N-LM)T v_j + LMT  \rho_j \\
&C_g(M, N) =  (N-LM)NT^2 LK_{max} \sum_j \mu_j v_j
 +  LMNT^2  \sum_j \lambda_j \rho_j + (LB_1 + B_2) NT.
\end{aligned}
\end{equation}
Since term $C_g(M, N)$ is independent of queue sizes,
by having $C_j(M,N)<0$ for all job types $j$,
the drift will be always negative for
large enough queues. We can ensure all $C_j(M,N)<0$ by choosing
\begin{equation}
\label{Nq-cond}
N > LM \max_{ j \in \mathcal{J}} \left(-1-\frac{\rho_j}{v_j}\right).
\end{equation}
Finally given any $\delta > 0$, we can ensure the Lyapunov drift (\ref{drift-cond}) is less than $-\delta$, if
\begin{equation}
\label{queue-cond}
\begin{aligned}
\min_j{C_j(M, N) Q_j(t_0)} < -\delta - C_g(M,N),
\end{aligned}
\end{equation}
which implies,
$
\max_{j} Q_j(t_0) > \frac{-\delta - C_g(M,N)}{\max_j C_j(M, N)},
$
or equivalently
$
\left \Vert \mathbf{q}_0 \right \Vert > \sqrt{J}
\frac{-\delta - C_g(M,N)}{\max_j C_j(M, N)}.
$

The proposition follows by choosing
$C_5 = L \max_{j \in \mathcal{J}} \left(-1-\frac{\rho_j}{v_j}\right)$
and $C_6(M, N, \delta) = \sqrt{J} \frac{ -\delta - C_g(M, N)}{\max_j C_j(M, N)}$.
\end{proof}
Therefore it follows that the Markov chain is positive recurrent
by the continuous-time version of Foster-Lyapunov theorem
and further the stability in the mean sense (\ref{eq:mean sense}) follows~\cite{meyn1993stability}. This concludes the proof of Theorem~\ref{theorem}.

\change{
\section{Generalizing arrival and service processes}\label{sec:generalization}
}
\change{
In Section \ref{sec-model}, we assumed Poisson arrivals and exponential service times. In this section, we show that our results in fact hold under much more general processes.
}

\change{
\subsection{Generalizing service time distribution}
\label{sec-gen-dist}
}
\change{
The assumption that service times follow exponential distribution
is not always realistic. Empirical studies in many applications suggest that service times
have heavy-tailed distributions~\cite{benson2010network,reiss2012heterogeneity}.
It is known that we can approximate a heavy-tailed distribution,
such as Pareto or Weibull, by using a hyper-exponential distribution, with high accuracy~\cite{Feldmann1998}. We
show that Theorem~\ref{theorem} still holds under hyper-exponential service time distributions.
}	
The probability density function
of hyper-exponential distribution is defined by
$
f(x) = \sum_{i=1}^n p_i \mu_i \exp{(-\mu_i x)},\ x \ge 0,
$
with $\sum_{i=1}^n p_i = 1$. This can be thought of as drawing a
value from $n$ possibly different exponential distributions and choosing one
of them with probability $p_i$, $i \in [1, \cdots, n]$.
The mean of the hyper-exponential is $\sum_{i=1}^n p_i \mu_{i}^{-1}$, while its variance is
$$
\left(\sum_{i=1}^n p_i \mu_{i}^{-1}\right)^2 +
\sum_{i=1}^n \sum_{j=1}^n p_i p_{j}
\left(\mu_{i}^{-1} - \mu_{j}^{-1} \right)^2.
$$
By choosing proper values of $p_i$ and $\mu_i$, we can generate
distributions that have the same mean as an exponential distribution with mean $\mu^{-1}$, but with variances much larger than $\mu^{-2}$ (variance of exponential distribution with mean $\mu^{-1}$).

\change{
Alternatively, we can view this as follows. Whenever a job is scheduled for service, it is assigned
to \textit{class} $i$ with probability $p_i$, $i \in \{ 1, \cdots, n \}$. A job of type $j$ that is in class $c$ will follow an exponentially distributed service time
with mean $\mu_{j, c}^{-1}$. By definition $\sum_{i=1}^n p_i \mu_{j, c}^{-1}= \mu_j^{-1}$ (where $\mu_j^{-1}$ is the mean service time for type-$j$ jobs as in the exponential case before).}
We then modify the definition of system state (\ref{eq-state}) to include the class of jobs in service. Specifically,
let $\mathcal{O}_j(t)$ be the set of all jobs of type $j$ being served at time $t$ in all the servers,
$\mathcal{O}_{j,\ell}(t)$ be those being served by server
$\ell$, and $c(i)\in \{1,\cdots, n\}$ denote the class of job $i \in \mathcal{O}_j(t)$.
We modify the Lyapunov function (\ref{lyap}) by considering that a scheduled job of type $j$ that is assigned to class $c$
will add a term $w_{j,c}$ to the queue size $Q_j$. The modified Lyapunov function is as follows
\begin{equation}\label{lyap-hyper}
	V(t) = \sum_j \frac{\left({Q}_j(t) +
	\sum_{i \in \mathcal{O}_j(t)} w_{j, c(i)} \right)^2}{2\mu_j}.
\end{equation}
Next we state the equivalent of Lemma~\ref{lem-inf} for the modified
Lyapunov function.

\begin{lemma}\label{lem-inf-hyper}
	By choosing
	\begin{equation} \label{eq-w-hyper}
	w_{j, c} = \frac{\mu_j}{\mu_{j, c}} - 1,
	\end{equation}
	the following bound holds at any time $t$:
	\begin{equation} \label{eq-inf-hyper}
	\begin{aligned}
	A V(t) \le \sum_j \Big[{Q}_j(t) \left(\rho_j -
	\sum_\ell I_\ell(t) \tilde{k}_j^\ell(t) +
	\sum_\ell(1-I_\ell(t))\frac{Ch}{\mu_j}\right)\Big] + B_h
	\end{aligned}
	\end{equation}
	where $C_h$ and $B_h$ are some constants.
\end{lemma}
\begin{proof}
See Appendix~\ref{proof:lem-inf-hyper} for the proof.	
\end{proof}
Using Lemma~\ref{lem-inf-hyper}, and redefining $\lambda_{max}$, $\mu_{max}$ and $T$
to include all types of jobs and all classes that a job can take,
the proof of Theorem~\ref{theorem} can be extended to the hyper-exponential distribution. We omit repeating the same arguments and mention the result as the following corollary.
\begin{corollary}
	\label{corol-hyper}
	Theorem~\ref{theorem} still holds if the service time
	distribution of jobs of type $j$ follows a hyper-exponential distribution
	with mean $\mu_j^{-1}$, $j \in \calJ$.
\end{corollary}
\begin{proof}
See Appendix~\ref{proof-corol-hyper} for the proof.
\end{proof}
\subsection{Batch arrivals}\label{sec-batch}
The Poisson assumption on the arrivals does not allow batch arrivals
at arrival events (only one job is added at any time). In practice,
however, a user may request multiple VMs simultaneously, or a Map job in a data-parallel cluster brings a set of tasks. To adapt our
model to such batch arrivals, we can consider a process where the
requests arrive at rate $\lambda$ and each arrival brings
a vector of VMs $\mathbf{v} = (v_1, \cdots, v_J)$ (i.e., $v_1$ VMs of
type $1$, $\cdots$, $v_J$ VMs of type $J$) with probability $p_{\mathbf{v}}$,
such that $\mathbf{v} \in \mathcal{V}$, for some bounded set $\mathcal{V}
\subset \mathds{N}_0^J$ and $\sum_{\mathbf{v} \in \mathcal{V}} p_{\mathbf{v}}=1$.
Theorem~\ref{theorem} can be extended to this setting. We state the extension as the following corollary.
\begin{corollary}
	\label{corol-batch}
	Suppose requests arrive as a Poisson process with rate $\lambda$, and each request brings a vector $\mathbf{v} = (v_1, \cdots, v_J) \in \mathcal{V}$ with probability $p_{\mathbf{v}}$.
	Define the workload of jobs of type $j$ as
	\begin{equation}\label{eq-batch-workload}
	\rho_j = \frac{\lambda \sum_{\mathbf{v} \in \mathcal{V}} v_j
		p_{\mathbf{v}}}{\mu_j},\ j \in \mathcal{J}.
	\end{equation}
	Under this new definition, Theorem~\ref{theorem} still holds.
\end{corollary}
\begin{proof}
See Appendix~\ref{proof:corol-batch} for the proof.
\end{proof}
Finally, it is also easy to verify that the arguments in Sections~\ref{sec-gen-dist} and~\ref{sec-batch} can be combined, to establish Theorem~\ref{theorem} under
both batch arrivals and hyper-exponential service distributions.
\section{Implementation Complexity and Customizations} \label{implementation}
Algorithm~\ref{alg1} described the basic non-permeative scheduling algorithm. In this 
section, we propose a few ways to customize the basic algorithm that might be more 
useful depending on the settings. For each suggestion, we briefly explain the 
advantages and discuss the implications in computational cost, as well as any possible 
modifications in the proof of the main theorem.
\subsection{Computing $\mathbf{r}$-max weight configuration} \label{sec-optconf}
Algorithm 1 assumes that there is a subroutine to compute an $r$-max weight configuration
when a job departs. In the case of $r=1$, the problem of finding the max weight
configurations is a hard combinatorial problem since it is an instance of
Knapsack problem \cite{Garey1979}; nevertheless there are approaches to solve this
problem in pseudo-polynomial time, or provide $r$-approximations ($r<1$)
in polynomial time~\cite{vazirani2013approximation,lin1998}. \change{Any $r$-approximation algorithm can be used in Algorithm~\ref{alg1} in a black box fashion.
Below, we briefly overview a few algorithms. The options discussed are not exhaustive 
and are only suggestive.}

\noindent \textbf{1. Finding max weight configuration ($r=1$)}

There are two approaches that are practically useful in this case:
\begin{enumerate}
\item[(i)] Each server can simply compute the set of its \textit{maximal}
configurations initially, i.e configurations in which no other extra job
can fit. This set has the same convex hull as $\mathcal{K}^\ell$
introduced in Section~\ref{sec-model} but it has significantly smaller
number of elements.
Every time, the max weight configuration is needed, server can search
only over the maximal configurations.
\item[(ii)] If the size of the server is large compared to the job sizes,
a dynamic programming approach is better. Assuming the maximum values
of the $R$ resource types of a server are $U_1, U_2, \cdots, U_R$,
the complexity of the algorithm is $O(J\times U_1 \times \cdots U_R)$
which is pseudopolynomial, but is still tractable assuming the number
of resource types is usually small (CPU, memory, disc, etc).
The dynamic programming approach requires to keep track of $G[\mathbf{u}]$
which is defined as the weight of the max weight configuration that
uses up to $\mathbf{u} = (u_1, \cdots, u_R)$ resources
($\mathbf{0} \le \mathbf{u} \le \mathbf{U}$).
Suppose $\mathbf{w}_j = (w_{j1}, \cdots w_{jR})$ is the resource requirement of
job $j \in \mathcal{J}$, then the dynamic programming recursion is as follows
\begin{equation}
G[\mathbf{u}] = \max_j\{G[\mathbf{u}-\mathbf{w}_j] + Q_j(t)\}, \nonumber
\end{equation}
with all values of $G$ being initially $0$.
%
%
\end{enumerate}
\textbf{2. Finding $r$-max weight configuration ($r<1$)\\}
There are several approximate algorithms to solve Knapsack, e.g., see~\cite{lin1998,vazirani2013approximation}. Below, we describe a simple greedy method.
\begin{lemma}
\label{lem2}
Consider a server $\ell$ with $R$ resource types. Suppose for every job type
$j \in \mathcal{J}$ we can fit at least $N_f\geq 1$ jobs of that type in the server.
If we only consider configurations that use one type of job and return the
one that gives the maximum weight, then the returned configuration will be
$r$-max weight configuration with $r=\frac{N_f}{R(N_f+1)}$.
\end{lemma}
\begin{proof}
See Appendix~\ref{proof-lem2} for the proof.
\end{proof}
Let $\mathbf{w}_j = (w_{j1}, w_{j2}, \cdots, w_{jR})$ be the vector of resource 
requirements of job type $j$, normalized with the the server capacity. Then, the simple 
greedy algorithm in Lemma~\ref{lem2} orders the job types according to their relative 
value, $Q_{j}(t)/(\max_n w_{jn})$, and fills the server with the job that has the 
maximum relative value. We can improve this greedy algorithm by iteratively scanning 
the job types with lower relative value and fitting the residual capacity of the server 
with these jobs, this should improve the performance in practice, however it does not 
change the theoretical result in Lemma~\ref{lem2} (which is a worst-case guarantee).

We notice that if $R \ge 2$ and $N_f=1$, the worst-case fraction of
the capacity region that Algorithm~\ref{alg1} provides, by using this greedy method as a subroutine, is small (at most $r^2$ fraction of the capacity region, due to requirement
$\beta<r$ in Theorem~\ref{theorem}). However, we can improve
Theorem~\ref{theorem}, as the the requirement $\beta<r$ can be relaxed to
$\beta<1$ in some cases, and Algorithm~\ref{alg1} can still achieve $r\beta$
fraction of the capacity region,
as stated in Corollary~\ref{prop_r} below.

\begin{corollary}\label{prop_r}
Consider a subset of configurations $\mathcal{\hat K}^\ell \subset \mathcal{K}^\ell$
and a subroutine that finds a max weight configuration out of this subset, i.e.
\[
{\mathbf{k}^\star}^{\ell}(t) = \arg\max_{\mathbf{k}^\ell \in
	\mathcal{\hat K}^\ell}f(\mathbf{k}^\ell, \mathbf{Q}(t)).
\]
Then Algorithm~\ref{alg1} that uses this subroutine to find an $r$-max weight configuration
and has parameter $\beta$, can support any workload vector $\bm{\rho}$
in the interior of $\mathcal{\hat C}_{\beta}$ which is the $\beta$ fraction of set
\be \label{cap-subset}
\mathcal{\hat C} = \{ \mathbf{x} \in \mathbb{R}^J_+ : \mathbf{x} = \sum_{\ell \in \mathcal{L}} \mathbf{x}^\ell,\ \mathbf{x}^\ell \in
Conv(\mathcal{\hat K}^\ell), \ell \in \mathcal{L}  \}
\ee
for $0<\beta<1$.
\end{corollary}
\begin{proof}
The proof exactly follows the proof of Theorem~\ref{theorem}, the only difference
is that now the capacity region is defined by a subset of all feasible configurations as in (\ref{cap-subset}).
\end{proof}
The implication of Corollary~\ref{prop_r} is that if $\mathcal{\hat C} \supset \mathcal{C}_r$ then
$\mathcal{\hat C}_{\beta} \supset \mathcal{C}_{r\beta}$
and the algorithm can support any workload vector  $\bm{\rho}$
in the interior of $\mathcal{C}_{r\beta}$ with $\beta < 1$.
This is indeed the case for the greedy algorithm of Lemma~\ref{lem2} as it
uses a subset of all the configurations (i.e., those with only one type of jobs).
\subsection{Customization of $\beta$}\label{sec:beta-custom}
\change{As explained, $\beta$ controls the trade off between throughput and delay. Higher $\beta$ makes a server stall more often, which increases the overall delay of jobs waiting to get service, however it
can achieve a higher long-run throughput. We notice that $\beta$ doesn't have to be constant, but can adapt to the queue size. Small queues can be a surrogate for low workload while large queues can indicate a high workload, thus by having $\beta$ automatically adapt to the queue sizes, we can avoid unnecessary stalling and achieve the best throughput-delay tradeoff}. In this section, we consider $\beta$ as a function of
$\mathbf{Q}$, as long as it converges to a desired value $\bar \beta$, when $\left \Vert \mathbf{Q} \right \Vert$ goes to infinity.
The following lemma states the main result.
\begin{corollary}
\label{lem-beta}
Suppose $\beta = h(\left \Vert  \mathbf{Q} \right \Vert_1)$
is an increasing function of
$\left \Vert  \mathbf{Q} \right \Vert_1 = \sum_j Q_j$
which satisfies the following:
$h(\mathbf{0}) = \bar\beta_{min}$ and
$\lim_{\left \Vert \mathbf{Q} \right \Vert_1 \to \infty}
h(\left \Vert \mathbf{Q} \right \Vert_1) = \bar \beta$
with $\bar \beta < r$. Then Algorithm 1 with this
queue-dependent $\beta$ can achieve $r \bar\beta$
fraction of the maximal throughput region $\mathcal{C}$.
\end{corollary}
\begin{proof}
See Appendix~\ref{proof:lem-beta} for the proof.
\end{proof}
As an example, a function that satisfies the requirements is
\begin{equation} \label{eq:beta}
	h(\mathbf{Q}) = \bar \beta (p + (1-p)\tanh(z \cdot \sum_j Q_j)),
\end{equation}
where
\begin{itemize}
\item $\bar \beta$ is the maximum value of the function and corresponds to the
fraction of capacity region that is achievable.
\item $z$ is the slope of sigmoid function at $0$ when $p=0$ which controls
how fast the function converges to the maximum value.
\item $p \in(-\infty, 1]$ is a constant that indicated how much constant value is weighted compared to
sigmoid function. $p=1$ makes function constant and equal to $\bar \beta$.
\end{itemize}
In simulations, we choose $p$ to be slightly less than $0$, and $z$ generally less than $0.01$, to avoid frequent configuration changes when the queue sizes are small.
The value of $\bar{\beta}$ depends on the long-run throughput (fraction of the capacity
region) that we want to achieve.
\subsection{Reducing stalled period duration}
\label{sec-redstall}
One way to reduce the stalled period duration further is to have a stalled server
transition to an active period, whenever the remaining jobs in the server are a
subset of the $r$-max weight configuration at that time (in addition to transition
at empty stalled times as before). Then, the server can become active
faster and renew its configuration according to the $r$-max weight configuration \textit{without any job preemptions}. The drawback is that more computation is needed, but this is
not a significant overhead given that servers will be most of the time active.
\subsection{Reducing configuration changes}
\label{sec-redchang}
An important problem with the proposed algorithm is that configuration changes
may happen very often and, approximately at the same time across the servers,
even with the suggested modification based on the queue-dependent $\beta$
(Section~\ref{sec:beta-custom}). The reason is that servers with the same
configuration will observe a similar queue vector, if any of their jobs finish
around the same time. This will make the condition~(\ref{st_cond}) either
true or false for all of these servers and will make most of them stalled before
any of them becomes active again. This behavior will continue if there is
no mechanism to stop it. To avoid this issue we can simply use the information of what fraction of
servers is stalled to decide whether to stall a server or not.
The modification that we suggest is to change the queue-dependent
$\beta$ to be $h(\mathbf{Q}(t)) \cdot q(s(t))$, where $s(t)\in [0,1]$ is the fraction
of servers which are stalled at time $t$ and $q$ is a decreasing function with $q(0)=1$.
To avoid having many servers getting stalled at the same time we need
the function $q$ to be very close to $0$ as $s$ approaches $1$. For example, it could be of the form $q(x) = \mathds{1}(x<p)$
to impose a hard limit of at most $p$ on the fraction of servers that can
be stalled at any time.

The proof arguments of Theorem~\ref{theorem} can be extended to this case.
The constant $B_1$ of Lemma~\ref{prop:weight_const}
can be modified to include the change in the queue sizes
when other servers are stalled. For this, one needs the estimate of
$M$ in Proposition~\ref{prop:Mq}. Another observation that simplifies the analysis is that our original proof treats all the servers
as stalled anyway when at least one of them is stalled so most of the arguments
of the original proof remains the same. We omit the detailed proof for brevity.



%
\section{Simulation Results}
\label{sec:simulations}
In this section, we verify our theoretical results and also compare the performance of our algorithm with two other algorithms, the randomized sampling algorithm~\cite{G2016} and the MaxWeight at local refresh
times~\cite{MS2014}, which will refer to them as G16 and M14 respectively (these algorithms were described in Section~\ref{sec-related}). We provide three sets of simulations using synthetic and real traffic traces:
(i) synthetic examples that our algorithm can handle effectively, while other algorithms fail, (ii) performance evaluation of algorithms with respect to the scaling
of the number of servers and scaling of traffic intensity, under both Poisson process and Log-normal inter-arrival times for the arrival process, and (iii) performance evaluation of algorithms using a real traffic trace from a large Google cluster.

Unless otherwise stated, our algorithm will have the following settings:
$r=1$, $\beta = h(\mathbf{Q}(t))q(s(t))$, for the $h$ function defined in
(\ref{eq:beta}) with $p=-0.05$, $z=0.005$, $\bar \beta=0.9$,
$q(s)=(1-s)\mathds{1}(s<0.1)$ where $s$ is the fraction of the stalled
servers at any time, as in Section~\ref{sec-redchang}.
Also the suggestion of Section~\ref{sec-redstall} is enabled.

Unless otherwise stated, the jobs arrive as a Poisson process
and service times are exponentially distributed as described in
Section~\ref{sec-model}, with the service times being independent
from job type and server. In case distributions of arrivals and
service times are different, we extend the definitions of $\lambda_j$
and $\mu_j$ from Section~\ref{sec-model} to be the mean number of
arrivals and the inverse of mean service time respectively, for each
job type $j$.
For each experiment we will also specify
the traffic intensity $\zeta \in (0, 1)$ of the workload. This
parameter controls how close the workload is to boundary of capacity
region $\mathcal{C}$. A workload $\bm{\rho}$ that has traffic intensity
$\zeta$ will therefore be on the boundary of the $\zeta$-fraction of the
capacity region $\mathcal{C}$.


\subsection{Inefficiency of other algorithms}\label{sec-ineff}
In this section we show handpicked examples where the other algorithms are either unstable or practically unusable,
yet our algorithm performs very well. For simplicity, we consider one dimensional case where there is one type of resource.

\textbf{Example 1 (Instability of M14: MaxWeight based on local refresh times).}
Consider one server with capacity $6$ units and two job types,
type-1 jobs require $4$ units and type-2 jobs require $1$ units.
Service rates are the same for both jobs and arrival rate of the
small job type is $8$ times higher than the large job type.
The traffic intensity is chosen to be $0.89$ so the workload vector is
$0.89\times (0.5, 4)$, which is clearly supportable because it is less than the average
of two maximal configurations $(1, 2)$ and $(0, 6)$. When the server starts scheduling
according to configuration $(1, 2)$, the arrival rate of small jobs
will be higher than their service rate.
That will result in the queue of small jobs to grow to infinity and
configuration never resets with a non-zero probability.
This will inevitably happen, since this probability exists every time
the server schedules according to configuration
$(1, 2)$. 
Figure \ref{badexpM14} depicts the total queue size (sum of the queue sizes) 
under our algorithm and M14. As it is seen, the queue sizes under M14 \cite{MS2014} go 
to infinity while Algorithm~\ref{alg1} keeps the queues stable. The sawtooth behavior 
under our algorithm in Figure \ref{badexpM14} indicates the configuration reset times.

\textbf{Example 2 (Large queue size under G16: Randomized sampling).}
In the second example we show that although G16 \cite{G2016}
guarantees stability it is possible that could yield very large queue sizes. Consider a relatively simple server setting as follows.
There are $4$ different types of servers with $1$, $2$, $4$, $8$ resource
units and $4$ types of jobs with resource requirements $1$, $2$, $4$, $8$
(thus each one can completely fill one of the servers).
Arrival and service rates are the same for all jobs and traffic load is $0.89$.
Figure \ref{badexp_G16} depicts the total queue size under the aggorithms. Intuitively one can see that this example
is hard for G16, since it can discover the best assignment to servers after
$4$ sampling events (one per queue) with probability $1/4^4=1/256$.
If there is a mistaken assignment, it is likely that it will lead to longer
waiting times for larger jobs that cannot fit in small servers.
\begin{figure}[t]
	\centering
	\begin{minipage}{.47\textwidth}
		\centering
		\includegraphics[width=\columnwidth,height=1.4 in]{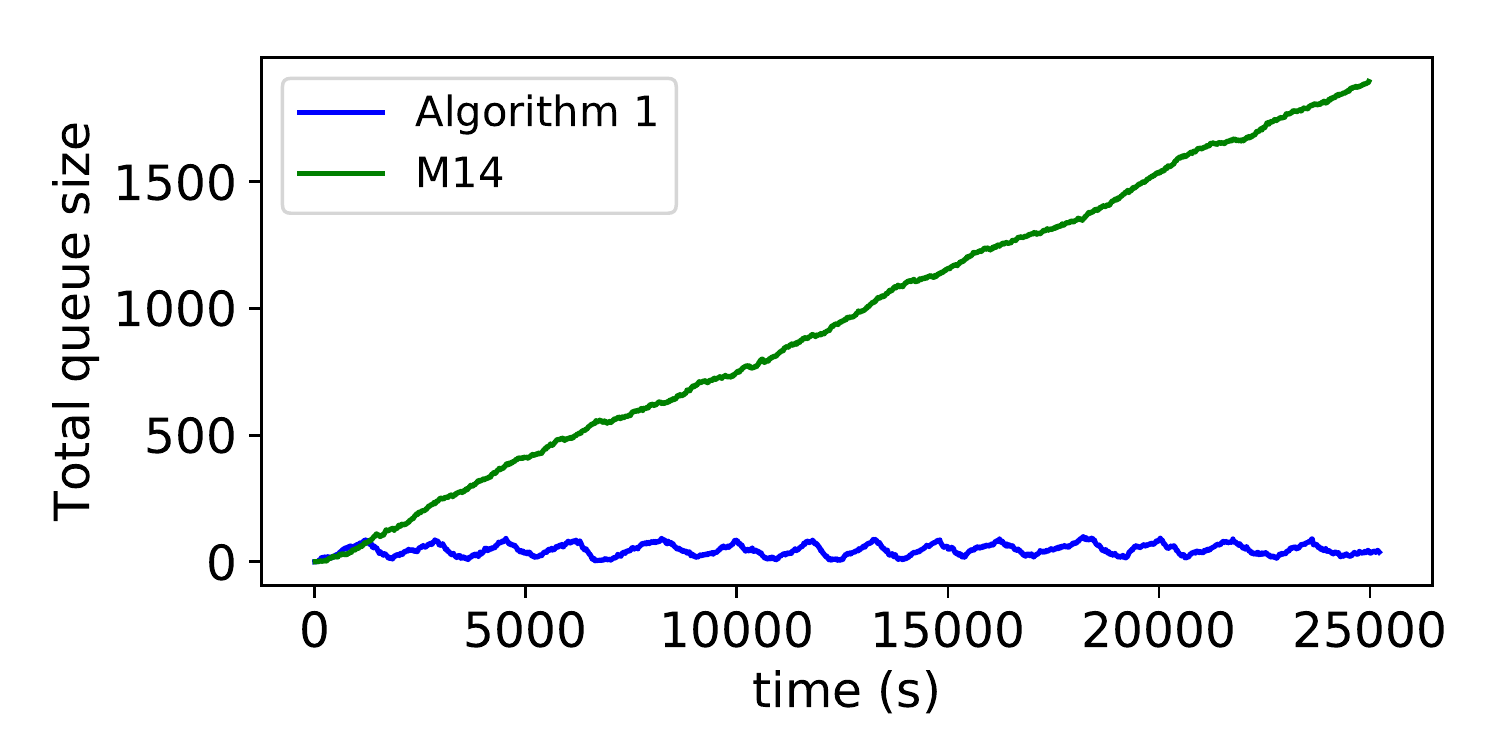}
		\caption{M14 fails in Example~1 while Algorithm~\ref{alg1} still stabilizes the queues.}
		\label{badexpM14}
	\end{minipage}\hfill
	\begin{minipage}{.52\textwidth}
			\centering
			\includegraphics[width=\columnwidth,height=1.4 in]{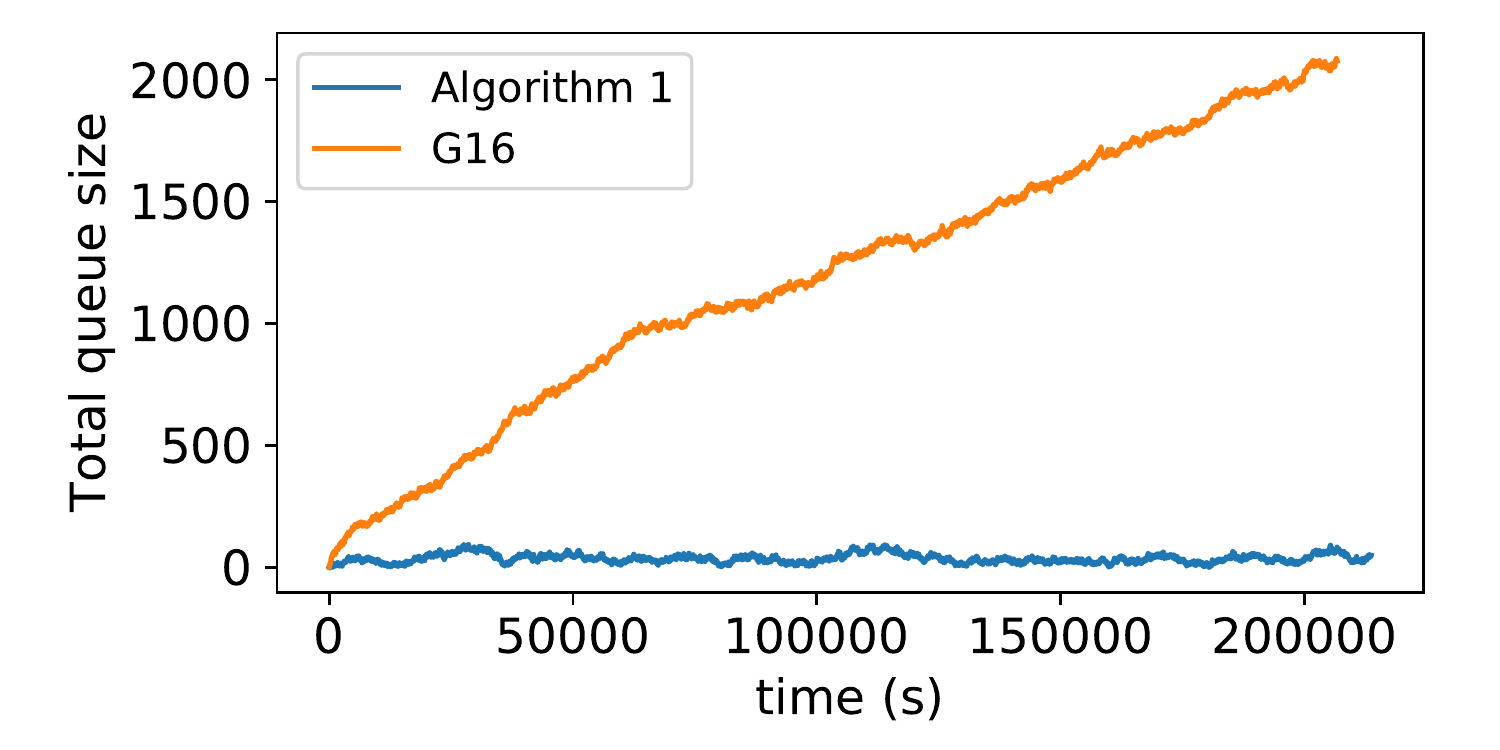}
			\caption{G16 performs poorly in Example~2 although it 
				theoretically converges. Algorithm~\ref{alg1} performs much better.}
			\label{badexp_G16}
	\end{minipage}\hfill
\end{figure}

\subsection{Scaling experiments}
\label{sec-scale-exp}
In this section, we use the VM types originally used in \cite{MSY2012,MS2014, G2016},
as indicated in Table \ref{tab-VM1}. In experiments, servers are homogeneous
with the capacities shown in Table \ref{tab-VM1}. All simulations were
repeated $5$ times and the results reported are the average of the $5$ runs.
For each run, we compute the time average of the total queue size which we refer to as \textit{the mean queue size} in the graphs. All algorithms were simulated for
$200000$ events except for G16 which was simulated for $400000$ events.
Events include arrivals and job completions, and in the case of G16,
they also include the sampling events of the queues.
In all cases we discarded the first $1/4$ fraction of the simulation traces
before computing the mean queue size of a run.

We perform all the simulations under two choices of inter-arrival time distributions:
Exponential (Poisson process) and Log-normal. The latter was used as empirical studies have shown that it is a good
model for the incoming traffic in datacenters~\cite{Ersoz2007}.

\begin{table}
\small
\centering
\caption{VM types and server types}
\begin{tabular}{|c|c|c|c|} \hline
 & Memory & CPU & Storage \\ \hline
Standard Instance & 15 GB & 8 EC2 & 1,690 GB \\ \hline
High-Memory Instance & 17.1 GB & 6.5 EC2 & 420 GB \\ \hline
High-CPU Instance & 7GB & 20 EC2 & 1,690 GB \\ \hline
\hline
Server & 90 GB & 90 EC2 & 5000 GB \\ \hline
\end{tabular}
\label{tab-VM1}
\end{table}

\textit{Scaling the number of servers.} We increase the number of servers to examine how well the algorithms scale. The number of servers ranges from 20 to 200.
The arrival rates were proportional to 
$[2/3, 11/3, 2/3]$
and scaled by the number of servers.
Service time distributions have the same mean for all job types
and are scaled such that the traffic intensity is 0.89.

Figure \ref{exp_server} shows the results of this experiment.
The behavior of Algorithm 1 and M14 is similar and they both perform
better as the number of servers increases, unlike G16.
As we can also see, the results are robust to the arrival process (Poisson vs Log-normal).

\begin{figure} [t]
\centering
	\begin{minipage}{.48\textwidth}
	\centering
	\includegraphics[width=\columnwidth]{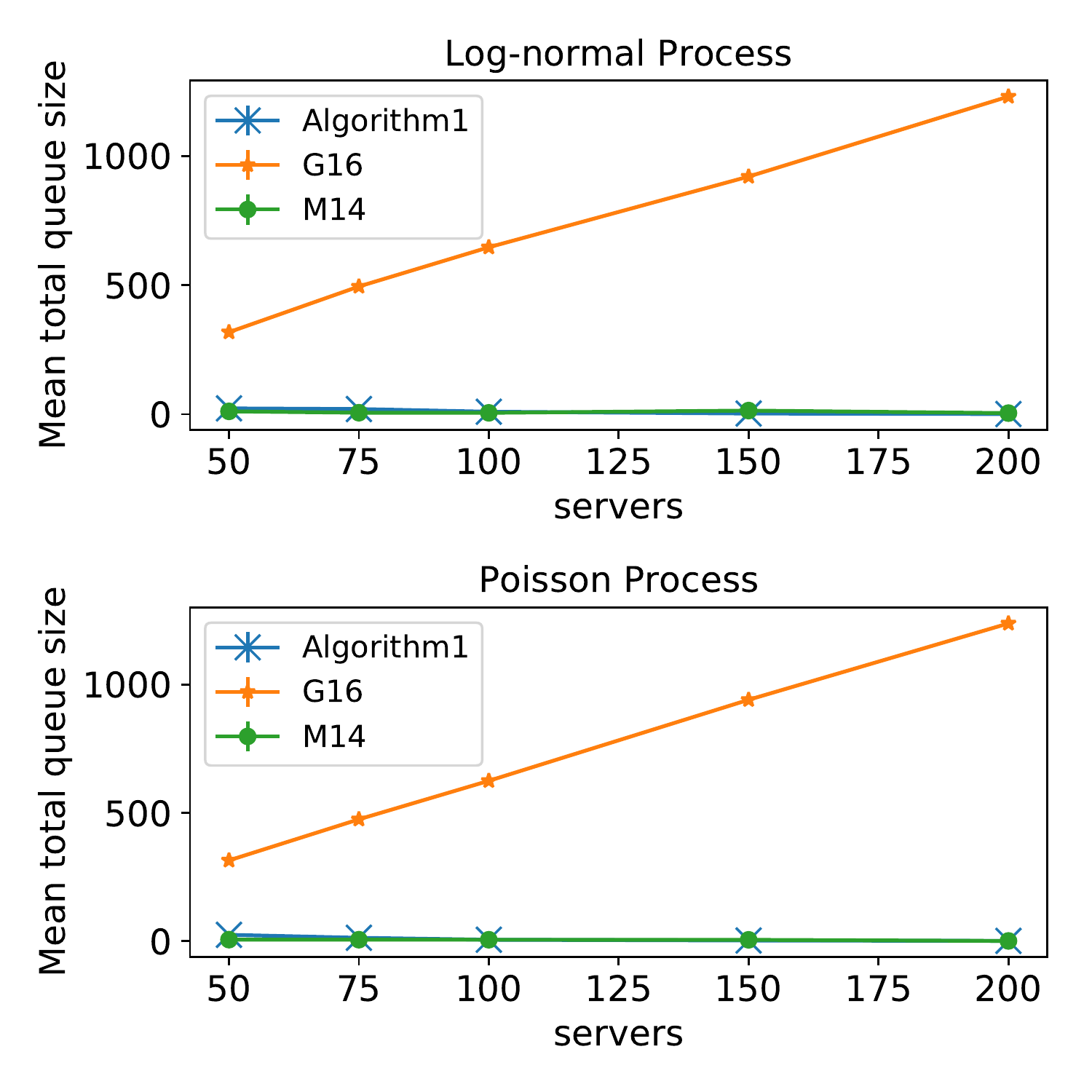}
	\caption{\change{Algorithm 1 is about as good as M14 and much better than G16
		when it comes to scaling cluster to more servers.}}
	\label{exp_server}
	\end{minipage}
\hfill
	\begin{minipage}{.48\textwidth}
	\centering
	\includegraphics[width=\columnwidth]{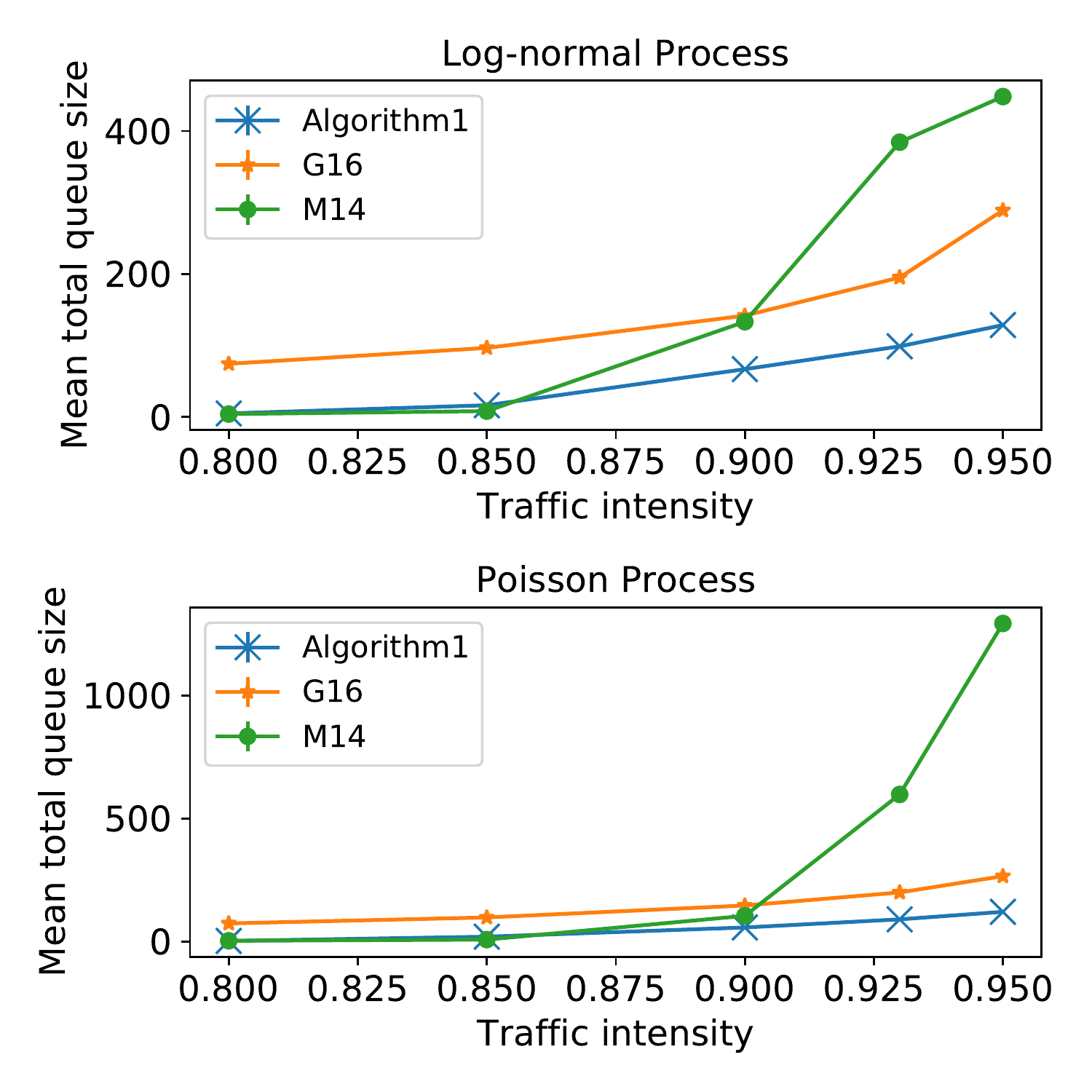}
	\caption{\change{Algorithm~\ref{alg1} has the most consistent performance. M14 deteriorates at higher 
		traffic and G16 deteriorates at lower traffic.}}
	\label{exp_workload}
	\end{minipage}
\end{figure}
\textit{Scaling the traffic intensity.} In the next experiment, we use the same server settings as before but now fix
the number of servers to $20$ and change the traffic intensity from $0.8$ to
$0.95$. To be consistent with our theoretical results, we choose
$\bar \beta=0.98$ in our algorithm so that it is higher than all the traffic intensities tested.
Arrival rates and departure rates are the same as before.

The results of this experiment are depicted in Figure~\ref{exp_workload}.
We notice that our algorithm performs very well in the whole range of workloads. The performance is also robust to the arrival process (Poisson vs Log-normal).
We can also see that M14 seems to become unstable in high traffic loads while G16 and Algorithm~1 are still stable.

\change{
\subsection{Experiment with Google trace dataset} \label{sec-realexp}
}

\change{
In this experiment, we use a real traffic trace from a large Google cluster,
to compare the performance in a more realistic setting. From the original dataset~\cite{ClusterData}, we extracted the arrival times of tasks
and their service times by taking the difference of the deployment time and the completion time. The trace characteristics are as follows: 
}
\change{
\begin{itemize}
	\item Trace includes two types of workload. One comes from batch tasks
		that are scheduled regularly and are not time critical and another
		comes from deployed user products that are serviced by long-running
		jobs \cite{Verma2015}. In our experiments, we extract only
		tasks that were completed without any interruptions,
		with their priority values being ignored.
	\item Resource requirements involve two resources (CPU and memory)
		and are collected once a job is submitted. The resources are not
		treated as discrete; their range in the original dataset
		is normalized to have a minimum of $0$ and a maximum of $1$ so
		they cannot be mapped directly into types. To map the jobs to
		a tractable number of types, we took the maximum out of the
		two resources and rounded it up to the closest integer power
		of $1/2$. All tasks that are mapped to the same power of two are
		considered to belong to the same type and will wait in the same
		queue. The highest power of $1/2$ considered was $7$, since
		lower valued jobs are very few and account for less than $1\%$
		of requests. The total number of queues is consequently $8$.
\item A total of about $18$ million jobs were extracted from trace after the above filtering.
The duration of the whole trace is $29$ days and the average job duration is
about half an hour. All findings about the trace are consistent with those
reported in \cite{Reiss2012} although there are some minor differences because
of the assumptions we made and the different way that the trace was processed.
	\item In actual trace the number of servers changes dynamically with
		servers being added, removed or modified. To keep things simpler
		we assumed that the sizes of all servers are all $1$ which is the
		maximum possible and their number is fixed throughout a run.
\end{itemize}
}

\change{
In the following simulations, we work with a window of $1$ million arrivals
which corresponds to approximately one and a half day. The traffic intensity
for that part of trace is depicted in Figure~\ref{exp_google_arrivals}, in terms of number of arrivals over 20-minute time intervals. The traffic intensity is variable and we suspect that the arrivals are correlated and do not really follow Poisson.
}
\begin{figure}[t]
	\centering
	\begin{minipage}{.45\textwidth}
		\centering
		\includegraphics[width=\columnwidth,height=0.6\columnwidth]{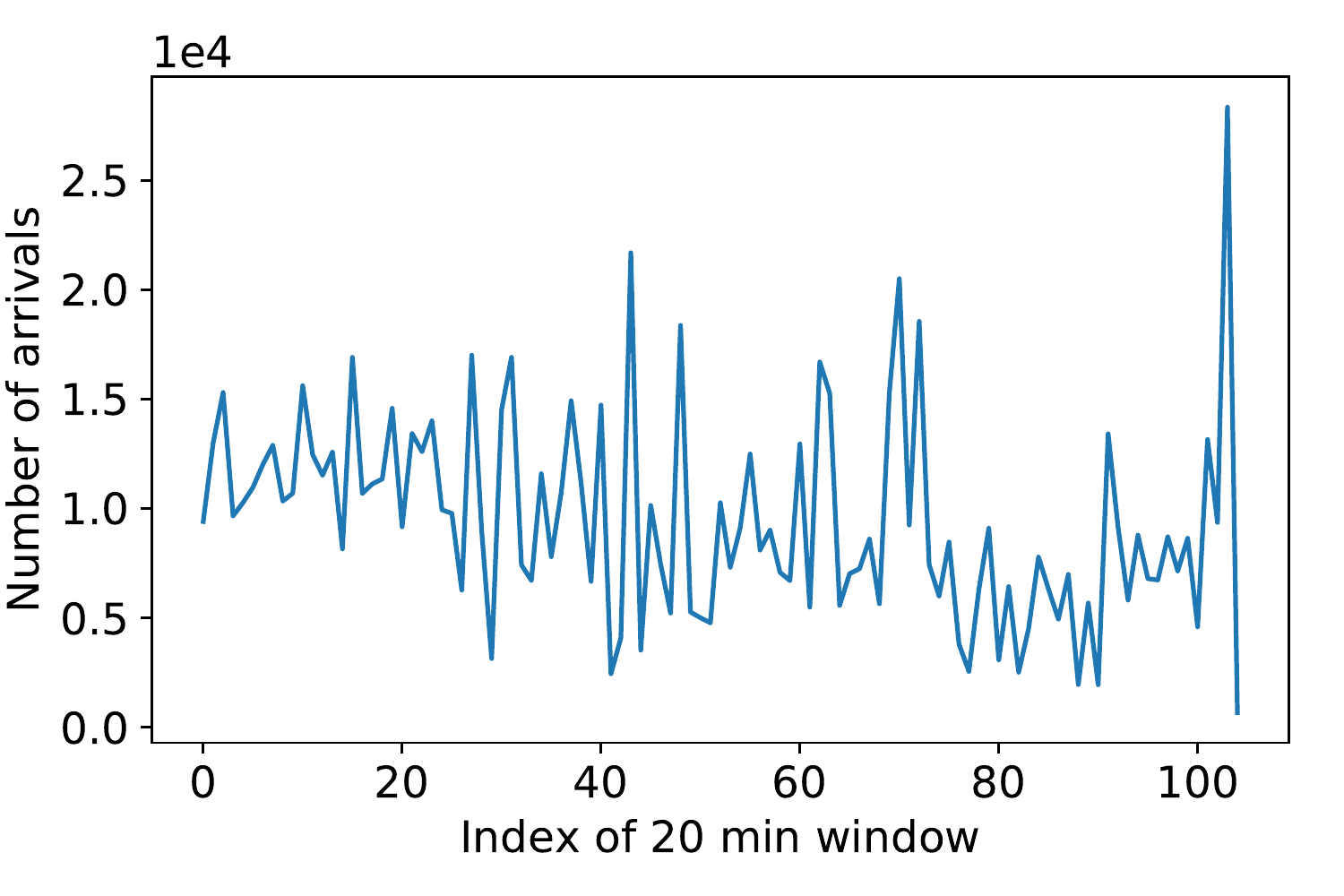}
		\caption{
			\change{
				Number of arrivals over time in the Google trace, computed over 20-minute time windows.
			}
		}
		\label{exp_google_arrivals}
	\end{minipage}\hfill
	\begin{minipage}{.5\textwidth}
			\centering
			\includegraphics[width=\columnwidth]{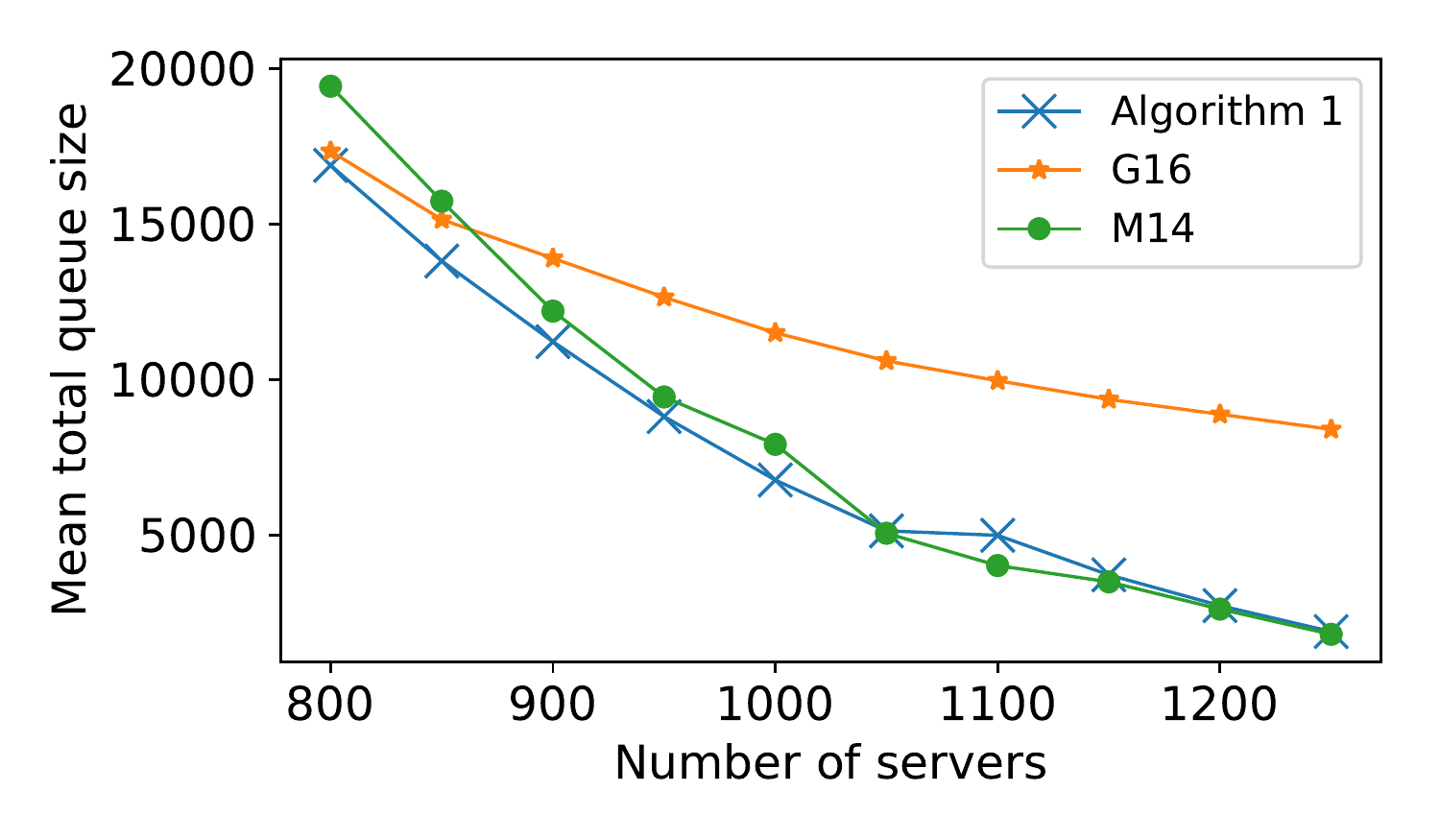}
			\caption{
				\change{The performance of different algorithms under the Google trace, for different number 
					of servers. 
				}
			}
			\label{exp_google1}
	\end{minipage}
\end{figure}

We evaluate the performance of all the algorithms using the above trace and for
different number of servers that ranges from $800$ to $1250$. Note that since
the trace is fixed and we have no control over it, the change in the number of servers
implicitly controls the traffic intensity.
All runs were repeated $3$ times and the reported results
which appear in Figure~\ref{exp_google1} is the average of these runs.
Our algorithm had the default configuration, with $z= 0.002$ and $q(s)=1-s$ if $s<0.015$ otherwise $q(s)=0$.
As we can see, our algorithm has the best overall performance in the whole range of the number of servers. The performance of G16 deteriorates as the number of servers scales up, while the performance of M14 deteriorates as the number of servers scales down, all consistent with our synthetic simulations.

\section{Conclusions}
In this paper, we introduced a new approach to non-preemptive VM scheduling in the cloud, with heterogeneous resources, and characterized the fraction of the maximum throughput that it can achieve. The algorithm can be tuned to provide a natural tradeoff between throughput, delay, and complexity. The evaluation results, using synthetic and real traffic traces, show that the algorithm outperforms the other methods, when the number of servers or the traffic intensity scales. In general, given an approximation algorithm to Knapsack with approximation ratio $r$, our algorithm can provide $\beta r$ fraction of the throughput region for $\beta <r$. One natural question is under which cases it is possible to relax this condition to $\beta <1$ 
(we saw it is indeed possible in the case of a greedy approximation algorithm). 
Other questions are related to how to incorporate preemptions 
(through proper preemption cost models), or provide deadline (strict delay) 
and fairness guarantees, which we postpone to future research.

\bibliographystyle{ACM-Reference-Format}
\bibliography{references,javad,javad2}

\appendix
\section{Proofs}\label{appendix-proofs}

\subsection{Proof of Proposition~\ref{prop:Mq}}\label{proof:Mq}

Let $\tau_j$ be the random variable denoting the service times of type-$j$ jobs.
In the proof, we use additionally the following notations:
the time when the state of server $\ell$ changes to active: $t^\ell_{(a)}$,
the duration that server remains active: $\Delta t^\ell_{(a)}$,
and the respective values when it changes to stalled state: $t^\ell_{(s)}$ and
$\Delta t^\ell_{(s)}$.

Let $K_{max}<\infty$ denote the maximum number of jobs that can fit in
any server, then at any time there are at most $LK_{max}$ jobs in all
the servers. A lower bound on the probability
$\mathds{P}(E_{\mathbf{S}(t_0),M,N})$ is then as follows
\begin{equation} \label{p1}
\begin{aligned}
\mathds{P}\left(E_{\mathbf{S}(t_0),M,N} \right) \ge^{(a)} &
\prod_\ell \prod_j
\mathds{P}\left(\tau_{j} < M T\right)^{\max(\bar{k}_j^\ell(t_0), 
	\tilde{k}_j^\ell(t_0))}
 \quad \mathds{P}\left(\Delta t^\ell_{(a)} > NT | \mathbf{S}(t_0) \right) \\
\ge^{(b)} & \left(1-e^{-M} \right)^{LK_{max}} 
 \prod_\ell \mathds{P}\left(\Delta t^\ell_{(a)} > NT | \mathbf{S}(t_0) \right).
\end{aligned}
\end{equation}
In the above, Inequality (a) bounds $\mathds{P}(E_{\mathbf{S}(t_0),M,N})$
by the probability that if a server becomes stalled at a time
in the interval $[t_0, t_0 + NT]$, it becomes empty within $MT$
amount of time, and once the server becomes active,
it remains active for at least $NT$ amount of time.
This ensures that a server will become stalled at most once in the interval
$[t_0, t_0 + NT]$ and for at most $MT$ time duration, as illustrated in
Figure \ref{fig-emn}.

Inequality (b) uses the fact that
$\mathds{P}(\tau_{j} < MT) \ge \mathds{P}(\tau_{j} < M/\mu_j) = (1-e^{-M})$,
since service time is exponentially distributed, and by bounding the maximum
number of jobs in system by $LK_{max}$.

To bound $\mathds{P}(\Delta t^\ell_{(a)} > NT | \mathbf{S}(t_0))$ in (\ref{p1}),
we use Lemma \ref{lem_active}.
Let $\mathbf{A}_1$ and $\mathbf{D}_1$ denote the arrival
and departure vectors respectively between the initial reference time $t_0$
and the first time server changes to active, $t^\ell_{(a)}$,
while $\mathbf{A}_2$ and $\mathbf{D}_2$
denote the same quantities between times $t^\ell_{(a)}$ and $t^\ell_{(a)} + NT$. For notational compactness, let $\mathbf{Q}(t_0)=\mathbf{q}_0$,
$\mathbf{Q}(t^\ell_{(a)})=\mathbf{q}_a^\ell$, $\mathbf{S}(t_0)=\mathbf{S}_0$.
Then
\change{
\be \label{p2}
\begin{aligned}
	& \mathds{P}\left(\Delta t^\ell_{(a)} > NT | \mathbf{S}_0\right)
	>^{(a)} \mathds{P}\left(\left \Vert \mathbf{A}_2 \right \Vert_\infty +
	\left \Vert \mathbf{D}_2 \right \Vert_\infty <
	C_b \left \Vert \mathbf{q}_a^\ell \right \Vert | \mathbf{S}_0 \right) \\
	& >^{(b)}  \mathds{P}\left(\left \Vert \mathbf{q}_a^\ell \right \Vert \ge
	C_c \left \Vert \mathbf{q}_0  \right \Vert \right)
	\mathds{P}\left(\left \Vert \mathbf{A}_2 \right \Vert_\infty +
	\left \Vert \mathbf{D}_2 \right \Vert_\infty <
	C_b \left \Vert \mathbf{q}_a^\ell \right \Vert  |
	\left \Vert \mathbf{q}_a^\ell \right \Vert \ge
	C_c \left \Vert \mathbf{q}_0 \right \Vert, \mathbf{S}_0 \right) \\
	& \ge \mathds{P}\left(\left \Vert \mathbf{q}_a^\ell \right \Vert \ge
	C_c \left \Vert \mathbf{q}_0  \right \Vert \right)
	\mathds{P}\left(\left \Vert \mathbf{A}_2 \right \Vert_\infty +
	\left \Vert \mathbf{D}_2 \right \Vert_\infty <
	C_b C_c \left \Vert \mathbf{q}_0 \right \Vert | \mathbf{S}_0 \right),
\end{aligned}
\ee
where $C_c$ is any arbitrary positive constant less than $1$
.}
In the above, Inequality (a) uses Lemma \ref{lem_active} with $C=C_b$,
and Inequality (b) is due to the law of total probability,
for the event $\left \Vert \mathbf{q}_a^\ell \right \Vert \ge
C_c \left \Vert \mathbf{q}_0  \right \Vert$ and its complement.

For notational compactness, let $\frac{C_b C_c}{2} = C_d$ and
$\frac{1-C_c}{2\sqrt{J}} = C_e$. Then the above probabilities
can be further bounded as follows:
\begin{equation} \label{p3}
\begin{aligned}
&\mathds{P}\left(\left \Vert \mathbf{A}_2 \right \Vert_\infty +
\left \Vert \mathbf{D}_2 \right \Vert_\infty <
C_b C_c \left \Vert \mathbf{q}_0 \right \Vert | \mathbf{S}_0 \right) \ge^{(a)} \\
&\mathds{P}\left(\left \Vert \mathbf{A}_2 \right \Vert_\infty <
C_d \left \Vert \mathbf{q}_0 \right \Vert | \mathbf{S}_0 \right)
\mathds{P}\left(\left \Vert \mathbf{D}_2 \right \Vert_\infty <
C_d \left \Vert \mathbf{q}_0 \right \Vert | \mathbf{S}_0 \right) =^{(b)} \\
& \prod_j \mathds{P}\left( A_{2,j} < C_d \left \Vert \mathbf{q}_0 \right \Vert
| \mathbf{S}_0 \right)
\mathds{P}\left(D_{2,j} < C_d \left \Vert \mathbf{q}_0 \right \Vert
| \mathbf{S}_0 \right) \ge^{(c)} \\
& \prod_j \left (1- \frac{\lambda_j NT}
{C_d \left \Vert \mathbf{q}_0 \right \Vert}\right)
\left (1- \frac{K_{max}\mu_j NT}{C_d
	\left \Vert \mathbf{q}_0 \right \Vert}\right) \ge^{(d)}
\left (1- \frac{\lambda_{max} NT}
{C_d \left \Vert \mathbf{q}_0 \right \Vert}\right)^J
\left (1- \frac{K_{max} \mu_{max} NT}
{C_d \left \Vert \mathbf{q}_0 \right \Vert}\right)^J,
\end{aligned}
\end{equation}
and
\begin{equation} \label{p4}
\begin{aligned}
& \mathds{P}\left(\left \Vert \mathbf{q}_a^\ell \right \Vert \ge
C_c \left \Vert \mathbf{q}_0 \right \Vert \right) =
\mathds{P}\left(\left \Vert \mathbf{q}_0 \right \Vert -
\left \Vert \mathbf{q}_a^\ell \right \Vert <
(1-C_c) \left \Vert \mathbf{q}_0 \right \Vert\right) \ge^{(e)} \\
& \mathds{P}\left(\left \Vert \mathbf{q}_0 - \mathbf{q}_a^\ell \right \Vert <
(1-C_c) \left \Vert \mathbf{q}_0 \right \Vert \right) =
\mathds{P}\left(\left \Vert \mathbf{A}_1 - \mathbf{D}_1 \right \Vert <
(1-C_c) \left \Vert \mathbf{q}_0 \right \Vert \right)  \ge^{(f)}  \\
& \mathds{P}\left(\left \Vert \mathbf{A}_1 \right \Vert
+ \left \Vert \mathbf{D}_1 \right \Vert <
(1-C_c) \left \Vert \mathbf{q}_0 \right \Vert \right) \ge^{(g)}
\mathds{P}\left(\left \Vert \mathbf{A}_1 \right \Vert_\infty +
\left \Vert \mathbf{D}_1 \right \Vert_\infty <
\frac{1-C_c}{\sqrt{J}} \left \Vert \mathbf{q}_0 \right \Vert\right) \ge^{(h)}  \\
& \left (1- \frac{\lambda_{max} NT}
{C_e \left \Vert \mathbf{q}_0 \right \Vert}\right)^J
\left (1- \frac{K_{max} NT^2 }{C_e \left \Vert \mathbf{q}_0 \right \Vert}\right)^J.
\end{aligned}
\end{equation}
In the above, (a) is due to the property $p(X+Y < C) \ge p(X < C/2) p(Y < C/2)$;
(b) is due to definition of infinity norm;
\change{
(c) is due to Markov's inequality
with arrival rates $\lambda_j$ independent of state $\mathbf{S}_0$
}
and departure rates upper-bounded by $K_{max} \mu_j^{-1}$,
also independent of state;
(d) uses that $\lambda_{max} \ge \lambda_j$ and
$T \ge 1/\mu_j$; (e) and (f) are due to triangle inequality; (g) is due to ratio
bound between infinity and 2-norm; and finally (h) is similar to (d).

Combining Equations (\ref{p1}), (\ref{p2}), (\ref{p3}), (\ref{p4}), we have
\[
\mathds{P}(E_{\mathbf{S}(t_0),M,N}) > \text{Factor}_0 \times \text{Factor}_1 \times \text{Factor}_2 \times \text{Factor}_3 \times \text{Factor}_4,
\]
where
\begin{equation} \label{p5}
\begin{aligned}
&\text{Factor}_0=\left(1-e^{-M} \right)^{ LK_{max}}, \quad \text{Factor}_1= \left(1-
\frac{\lambda_{max} NT}{C_d \left \Vert \mathbf{q}_0 \right \Vert}\right)^{LJ}, \quad
\text{Factor}_2=
\left (1- \frac{K_{max} NT \mu_{max}}{C_d \left \Vert \mathbf{q}_0 \right \Vert}\right)^{LJ}, \\
& \text{Factor}_3=\left (1- \frac{\lambda_{max} NT}{C_e \left \Vert \mathbf{q}_0
	\right \Vert}\right)^{LJ},\quad
\text{Factor}_4=
\left (1- \frac{K_{max} NT \mu_{max} }{C_e \left \Vert \mathbf{q}_0 \right \Vert}\right)^{LJ}.
\end{aligned}
\end{equation}
Hence, to ensure $\mathds{P}(E_{\mathbf{S}(t_0),M,N}) > 1 - \epsilon$, 
it suffices that each of the 5 factors, $\text{Factor}_0$, $\text{Factor}_1$, 
$\text{Factor}_2$, $\text{Factor}_3$, $\text{Factor}_4$, to be greater than 
$(1-\epsilon)^{1/5}$.

Using the inequality $(1-c)^x > 1-cx$ for $x>1$, it is sufficient to have
\begin{equation*} \label{Mq-cond}
	\begin{aligned}
		M &> \log \left(\frac{5L K_{max}}{\epsilon} \right); \left \Vert \mathbf{q_0} \right \Vert
		> \frac{ 5LJNT \max(\lambda_{max}, K_{max} T)}{\epsilon \min(C_d, C_e)}.
	\end{aligned}
\end{equation*}
Finally the Proposition follows if
$C_1 = \log \left(5 LK_{max}\right)$ and
$C_2 = 5LJT \max(\lambda_{max}, K_{max} T)/{\min(C_d, C_e)}$.

\subsection{Proof of Lemma~\ref{lem-inf}} \label{proof:lem-inf}
Note that what we want to bound is the following expression and
then take its limit as $u$ goes to $0$.
\ben
\frac{\mathds{E}_{\mathbf{S}(t)}[V(t+u)]-V(t)}{u} = \sum_j
\frac{\mathds{E}_{\mathbf{S}(t)}[Q_j(t+u)^2 - Q_j(t)^2]}{2u \mu_j}.
\een
By definition,
\[{Q}_j(t+u) = {Q}_j(t) + A_j(t, t+u) - D_j(t, t+u),\] where $A_j(t, t+u)$
and $D_j(t, t+u)$ are respectively the number of arrivals and departures
of type $j$ from $Q_j$ during $(t, t+u)$. By squaring the both sides,
it is straightforward to see that
\begin{equation}
\begin{aligned}
Q_j(t+u)^2 \leq Q_j(t)^2 + A_j(t, t+u)^2+D_j(t, t+u)^2 + 2Q_j(t)(A_j(t, t+u) - D_j(t, t+u)). \nonumber
\end{aligned}
\end{equation}
Recall that number of arrivals is a Poisson process with rate $\lambda_j$
and each job $j$ already in a server leaves after an exponentially
distributed amount of time with rate $\mu_j$. Hence, it is easy to see that
\begin{equation} \label{eq-Asqr}
\begin{aligned}
&\mathds{E}_{\mathbf{S}(t)}[A_j(t, t+u)^2] =
\lambda_j u  + \oo(u),
\end{aligned}
\end{equation}
and similarly for $D_j(t, t+u)$,
\begin{equation} \label{eq-Dsqr}
\begin{aligned}
\mathds{E}_{\mathbf{S}(t)}[D_j(t, t+u)^2] \le &\sum_{\ell} I_\ell(t)
{\tilde{k}}^{\ell}_j(t)\mu_j u +
\sum_{\ell} (1-I_\ell(t))\sum_{j^\prime}
{\bar{k}}^{\ell}_{j^\prime}(t)\mu_{j^\prime} K_{max}^2 u
+ \oo(u)\\
\le & L K_{max} \mu_j u + LK_{max}^3 \mu_{max} u+\oo(u).
\end{aligned}
\end{equation}
In the above bound, we used the fact that a job $j$ may depart from queue
either when a job $j$ completes service in an active server or when any
job departs from a stalled server and makes the server empty, in which case
up to $K_{max}$ jobs can be scheduled in that server.
We also used that $\sum_j\sum_\ell k^{\ell}_j \le LK_{max}$ and
that $\mu_j \le \mu_{max}$ for any job type $j$.

Assuming $Q_j(t)>0$, if server $\ell$ is in an active period then
${\bar{k}}^{\ell}_j(t)={\tilde{k}}^{\ell}_j(t)$
(i.e., there are no empty slots for type-$j$ jobs).
and the above inequality also clearly holds if $Q_j(t)=0$.
Using the the indicator function $I_\ell(t)$,
we can write the following inequality that holds
for any state of servers.
\begin{equation} \label{eq-difQsqr}
\begin{aligned}
\mathds{E}_{\mathbf{S}(t)}[Q_j(t+u)^2 - Q_j(t)^2] \leq
\lambda_ju + LK_{max}\mu_j u + LK_{max}^3 \mu_{max} u + 
2Q_j(t)(\lambda_j-\sum_{\ell} I_\ell(t) {\tilde{k}}^{\ell}_j(t)\mu_j)u
+\oo(u).
\end{aligned}
\end{equation}
Notice that in the above upper bound, we have ignored the queue departures when
the server is in a stalled period.

Thus at any time $t$, taking the limit as $u \to 0$,
\begin{equation} \label{eq-AV-limit}
\begin{aligned}
A V(t) \le  \left[\sum_j {Q}_j(t)\left(\rho_j - \sum_\ell
I_\ell(t) \tilde{k}_j^\ell(t) \right) \right]+B_2,
\end{aligned}
\end{equation}
for a constant $B_2=\sum_j(\rho_j+LK_{max}^3\frac{\mu_{max}}{\mu_j} + LK_{max})$.

\subsection{Proof of Lemma~\ref{prop:weight_const}}\label{proof:weight_const}
Define $R_0(\tilde{\mathbf{k}}^\ell)$ as
the set of queue size vectors $\mathbf{Q}$ for which
$f(\tilde{\mathbf{k}}^\ell, \mathbf{Q}) >
\beta r f(\mathbf{k}^\ell, \mathbf{Q})$ for any
$\mathbf{k}^\ell \in \mathcal{K}^\ell$. Similarly define
$R_1(\tilde{\mathbf{k}}^\ell)$ as the set of queue size vectors not in
$R_0(\tilde{\mathbf{k}}^\ell)$ for which
$f(\tilde{\mathbf{k}}^\ell, \mathbf{Q}) >
\beta r f(\mathbf{k}^\ell, \mathbf{Q}) - B_1$
for any $\mathbf{k}^\ell \in \mathcal{K}^\ell$
and finally $R_2(\tilde{\mathbf{k}}^\ell)$
as the set of the queue size vectors not in
$R_0(\tilde{\mathbf{k}}^\ell)$ or $R_1(\tilde{\mathbf{k}}^\ell)$.
We want to show that, with high probability, the queue size vector
does not take a value in $R_2(\tilde{\mathbf{k}}^\ell)$ during
an active period.
%
%

Note that at the beginning of an active period, the queue size vector
is in the set $R_0(\tilde{\mathbf{k}}^\ell)$ and the active period of
server $\ell$ ends when
at \textit{the time of a job departure from server $\ell$},
the queue size vector is either in $R_1(\tilde{\mathbf{k}}^\ell)$ or
$R_2(\tilde{\mathbf{k}}^\ell)$. Let $t_{i}$ be the $i$-th time that
the queue size vector transitions from set $R_0(\tilde{\mathbf{k}}^\ell)$
to $R_1(\tilde{\mathbf{k}}^\ell)$ while still in the active period.
Then there are three possible cases after $t_i$:

1. the queue size vector transitions back to $R_0(\tilde{\mathbf{k}}^\ell)$
before a job departs from server $\ell$,

2. the queue size vector remains in $R_1(\tilde{\mathbf{k}}^\ell)$
until a job departs from server $\ell$,

3. the queue size vector reaches $R_2(\tilde{\mathbf{k}}^\ell)$
before next job departure from server $\ell$.	

We denote the respective probabilities that each the events above
occurs by $p_0(t_i)$, $p_1(t_i)$ and $p_2(t_i)$.

The event $E_{B_1, \ell}$, which according to description
is the event that $f(\tilde{\mathbf{k}}^\ell, \mathbf{Q}(t)) >
\beta r f(\mathbf{k}^\ell, \mathbf{Q}(t)) - B_1$,
for any $\mathbf{k}^\ell \in \mathcal{K}^\ell$ and at
any time $t$ in the active period,
does not occur with probability
\begin{equation} \label{eq-PEB1}
1 - \mathds{P}(E_{B_1}, \ell) = \sum_{i=0}^{\infty}
p_2(t_{i})\prod_{j=0}^{i-1} p_0(t_{j}),
\end{equation}
which we want to show it is less than
$\epsilon$ for $B_1$ large enough.

First note that $p_0(t_{i})$ is strictly less
than $1$, i.e., $p_0(t_{i}) < 1 - C_f$ for some positive constant $C_f$.
To see that, note that the second case will occur,
if the next event after time $t_{i}$ is a job departure
from server $\ell$. Arrival and service processes are all
Poisson and the rate of both is at most
$r_1 = {J\lambda_{max} + LK_{max} \mu_{max}} \,.$	
The rate of job departures from server $\ell$ is also Poisson
and has a rate of at least $r_2 = \min_{j \in \mathcal{J}} \mu_j\,.$
The probability that departure from server $\ell$ happens before
any other event is therefore at least $C_f = \frac{r_2}{r_1 + r_2}$
and hence $p_1(t_{i}) \ge C_f \,$ and consequently
\begin{equation}\label{eq-p0Cf}
p_0(t_{i}) < 1 - p_1(t_{i}) < 1 - C_f ,\quad C_f=\frac{r2}{r1+r2}.
\end{equation}

Next we find an upper bound on $p_2(t_{i})$. At every arrival or
departure each of the queue sizes can change by at most $K_{max}$.
Considering $t$ is the time that the queue change occurs, and
$t^-$ the time right before	the change, the change in the weight
of the server configuration can be bounded as
\begin{equation}
\begin{aligned}
&f(\mathbf{k}^\ell, \mathbf{Q}(t)) -
f(\mathbf{k}^\ell, \mathbf{Q}(t^-)) = 
\sum_j k^\ell_j (Q_j(t) - Q_j(t^-)) \le
K_{max} \sum_j k^\ell_j \le	K_{max}^2.
\end{aligned} \nonumber
\end{equation}
The difference between configuration weights of any two
queue size vectors, with one in the set $R_0(\tilde{\mathbf{k}}^\ell)$
and the other in $R_2(\tilde{\mathbf{k}}^\ell)$, is at least
$B_1$ by definition. Therefore the number of events (arrivals
or departures) needed to transition from one set to the other
is at least $N_{B_1} = \lceil\frac{B_1}{K_{max}^2}\rceil$
and they should occur before any departure from server
$\ell$. The probability that this happens is
$(1 - C_f)^{N_{B_1}-1}$ for the choice of $C_f$ in~(\ref{eq-p0Cf}).
The time $t_{i}$ is the time that the first of these
events happens, which makes the queue size vector transition
to set $R_1(\tilde{\mathbf{k}}^\ell)$, hence
\begin{equation} \label{eq-p2B1}
p_2(t_{i}) \le (1- C_f)^{N_{B_1} - 1}.
\end{equation}
Lastly using Inequalities (\ref{eq-p0Cf}) and (\ref{eq-p2B1}) in
(\ref{eq-PEB1}), we get
\begin{equation}
1 - \mathds{P}(E_{B_1}, \ell) < \frac{(1-C_f)^{N_{B_1} -1}}{C_f} \nonumber.
\end{equation}
We can ensure that this expression is less than $\epsilon$ by choosing
$B_1 > - C_3 \log{\epsilon} + C_4,$ where the constants $C_3$ and $C_4$ are
\[
C_3 = -\frac{K_{max}^2}{\log{(1-C_f)}},\ C_4 = \frac{K_{max}^2\log{C_f}}{\log{(1-C_f)}}.
\]

\change{
\subsection{Proof of Lemma~\ref{lem-inf-hyper}} \label{proof:lem-inf-hyper}
}

\change{
	Following the steps of Lemma~\ref{lem-inf} we will first find a bound
	for the change in the nominator of the Lyapunov function in an interval
	$[t, t+u]$, for a particular job type $j$.
	State $\mathbf{S}(t)$ is defined as in Section~\ref{sec-lyap}
	but now it also includes the classes of the scheduled jobs
	$\mathcal{O}_j(t)$ for every $j \in \mathcal{J}$.
	Throughout the proof we will use that values $w_{j, c}$ are bounded
	or more specifically that $W = \max_{j,c} |w_{j, c}| < \infty$.
}

\change{
	Using the definition of Equation~(\ref{lyap-hyper}) we get
	\begin{equation}
	\begin{aligned} \label{eq-Qsqr-hyper}
	&\mathds{E}_{\mathbf{S}(t)}
	\Bigg[\left(Q_j(t+u) + \sum_{i \in \mathcal{O}_j(t+u)}
	w_{j, c(i)} \right)^2 - \left(Q_j(t) +
	\sum_{i \in \mathcal{O}_j(t)} w_{j, c(i)} \right)^2 \Bigg] \leq\\
	&\mathds{E}_{\mathbf{S}(t)}
	\Bigg[Q_j(t+u)^2 - Q_j(t)^2\Bigg] +
	\mathds{E}_{\mathbf{S}(t)}
	\Bigg[\left(\sum_{i \in \mathcal{O}_j(t+u)} w_{j, c(i)}\right)^2 -
	\left(\sum_{i \in \mathcal{O}_j(t)} w_{j, c(i)}\right)^2\Bigg] + \\
	&2\mathds{E}_{\mathbf{S}(t)}
	\Bigg[Q_j(t+u)\sum_{i \in \mathcal{O}_j(t+u)} w_{j, c(i)}
	-Q_j(t)\sum_{i \in \mathcal{O}_j(t)} w_{j, c(i)}\Bigg]. \\
	\end{aligned}
	\end{equation}
	Now we will give bounds for each one of the above terms
	and we will combine them later.
}

\change{
	The first one can be bounded with the same approach as the one that
	gave the bound of Equation~(\ref{eq-difQsqr}).
	The only difference here is that each job has different service rate that
	depends on its state and $\mu_{max}$ is now equal to $\max_{j,c} \mu_{j,c}$.
	The bound we get is then
	\begin{equation}
	\begin{aligned}
	\label{eq-hyper-term1}
	&\mathds{E}_{\mathbf{S}(t)}[Q_j(t+u)^2 - Q_j(t)^2] \leq \\
	&\lambda_ju
	+ LK_{max}\mu_j u + LK_{max}^3 \mu_{max} u +
	2Q_j(t)\left(\lambda_j-\sum_{\ell} I_\ell(t) \sum_{i \in \mathcal{O}_j(t)}
	\mu_{j, c(i)}\right)u +\oo(u).
	\end{aligned}
	\end{equation}
}

\change{
	For the second one we rely on the fact that the expression
	$\left(\sum_{i \in \mathcal{O}_j(t)} w_{j, c(i)}\right)^2$ is
	between $0$ and $(LK_{max}W)^2$ and that is the largest change
	that can take place.
	Of course we also need to use the rate at which this change
	occurs in an interval of length $u$,
	which is at most $\lambda_j + LK_{max}\mu_{max}$.
	The result will be the following inequality:
	\begin{equation}
	\begin{aligned}
	\label{eq-hyper-term2}
	\mathds{E}_{\mathbf{S}(t)}
	\Bigg[\left(\sum_{i \in \mathcal{O}_j(t+u)} w_{j, c(i)}\right)^2 -
	\left(\sum_{i \in \mathcal{O}_j(t)} w_{j, c(i)}\right)^2\Bigg] \le
	(LK_{max}W)^2(\lambda_j + LK_{max}\mu_{max})u.
	\end{aligned}
	\end{equation}
}

\change{
	Lastly we can break the last expectation term in two parts using the fact that
	${Q}_j(t+u) = {Q}_j(t) + A_j(t, t+u) - D_j(t, t+u)$. The first part is
	proportional to $Q_j(t)$ and the latter is bounded since
	expected arrivals and departures are bounded.
	Notice that the expected value of weight of newly scheduled jobs is
	$\sum_{c=1}^S p_c w_{j, c} = 0$, so only the jobs that depart are considered
	in first term. Again the result is the following:
	\begin{equation}
	\begin{aligned}
	\label{eq-hyper-term3}
	&2\mathds{E}_{\mathbf{S}(t)}
	\Bigg[Q_j(t+u)\sum_{i \in \mathcal{O}_j(t+u)} w_{j, c(i)}
	-Q_j(t)\sum_{i \in \mathcal{O}_j(t)} w_{j, c(i)}\Bigg] \le \\
	&2\mathds{E}_{\mathbf{S}(t)}\Bigg[Q_j(t)
	\left(\sum_{i \in \mathcal{O}_j(t+u)\setminus  \mathcal{O}_j(t)} w_{j, c(i)}
	- \sum_{i \in \mathcal{O}_j(t)\setminus
		\mathcal{O}_j(t+u)} w_{j, c(i)} \right)\Bigg] + \\
	&2\mathds{E}_{\mathbf{S}}\Bigg[\left(A_j(t, t+u) - D_j(t, t+u)\right)
	\sum_{i \in \mathcal{O}_j(t+u)} w_{j, c(i)}\Bigg] \le \\
	&2 Q_j(t) \left(-\sum_{i \in \mathcal{O}_j(t)}
	\mu_{j, c(i)} w_{j, c(i)} \right)u +
	LK_{max} W(\lambda_j + LK_{max}\mu_{max})u.
	\end{aligned}
	\end{equation}
}

\change{
	Putting together Equations (\ref{eq-hyper-term1}), (\ref{eq-hyper-term2})
	and (\ref{eq-hyper-term3}) we get:
	\begin{equation} \label{eq-hyper-final}
	\begin{aligned}
	&\mathds{E}_{\mathbf{S}(t)}
	\Bigg[\left(Q_j(t+u) + \sum_{i \in \mathcal{O}_j(t+u)}
	w_{j, c(i)}(t+u) \right)^2 -\left(Q_j(t) + \sum_{i \in \mathcal{O}_j(t)}
	w_{j, c(i)}(t) \right)^2 \Bigg] \\
	&\leq
	B_h u + 2Q_j(t)(\lambda_j-\sum_{\ell} I_\ell(t)
	\sum_{i \in \mathcal{O}^\ell_j(t)} (1 + w_{j, c(i)})\mu_{j, c(i)}
	-\sum_{\ell}(1-I_\ell(t))
	\sum_{i \in \mathcal{O}^\ell_j(t)} \mu_{j, c(i)} w_{j, c(i)}(t)),
	\end{aligned}
	\end{equation}
	where $B_h = \lambda_j
	+ LK_{max}\mu_j + LK_{max}^3 \mu_{max} + (LK_{max}W)^2(\lambda_j + LK_{max}\mu_{max}) + LK_{max}W(\lambda_j + LK_{max}\mu_{max})$.
}

\change{	
	Finally by applying the definition of $A V(t)$ from
	equation~(\ref{eq-infA}) to (\ref{eq-hyper-final}) and substituting
	$(1 + w_{j, c(i)})\mu_{j, c(i)}$ with $\mu_j$
	-- as implied by definition~(\ref{eq-w-hyper}) -- and
	$\sum_{i \in \mathcal{O}^\ell_j(t)} \mu_{j, c(i)} w_{j, c(i)}(t)$ by
	its upper bound $K_{max}W \mu_{max}$, we get the result of the lemma,
	for $C_h = K_{max}W \mu_{max}$.
}

\change{
\subsection{Proof of Corollary~\ref{corol-hyper}}\label{proof-corol-hyper}
Essentially Lemma~\ref{lem-inf-hyper} shows that the infinitesimal generator
can be bounded similar to (\ref{lem-inf}) for exponential distribution, 
with only one extra term: $\sum_j Q_j(t) \sum_\ell(1-I_\ell(t))\frac{C_h}{\mu_j}$,
which is nonzero only if there is at least one stalled server at time $t$. 
However we know, the total cumulative time duration that there are any stalled servers, 
is at most $LMT$ by Proposition~\ref{prop:Mq} (the same arguments hold). 
As a result, in the proof of Proposition~\ref{prop:Nq}, we only need to change 
the second term of Equation (\ref{eq:drift2}) to
\[
LMT\mathds{E}_{\mathbf{S}(t_0)}
\left[
\max_t{
	\sum_j Q_j(t)\left(\rho_j +\frac{LC_h}{\mu_j} \right)
}
\right]
\]
and ultimately constant $C_5$ of final result to
$
C_5 = L\max_{ j \in \mathcal{J}}
\left(-1-\frac{\rho_j +\frac{LC_h}{\mu_j}}{v_j}\right).
$
}

\change{
\subsection{Proof of Corollary~\ref{corol-batch}}\label{proof:corol-batch}
There are three parts in the original proof that need to change if
we redefine the arrival rate of a job type $j$ as
\begin{equation}
\lambda_j =
\lambda \sum_{\mathbf{v} \in \mathcal{V}} v_j p_{\mathbf{v}}
\end{equation}
and the workload of a job type $j$ as in Equation~(\ref{eq-batch-workload}).
}

\change{
The first change to the previous proof (under Poisson assumption) is to modify the bound of Equations (\ref{p3})
and (\ref{p4}) since they relied on the assumption that arrivals are independent, whereas under the batch arrivals, the arrivals of various job types are no longer independent.
We can still compute a new bound as follows
\begin{equation} \label{p3_dep}
\begin{aligned}
\mathds{P}\left(\left \Vert \mathbf{A}_2 \right \Vert_\infty <
C_f \left \Vert \mathbf{q}_0 \right \Vert \right) \ge
\mathds{P}\left(\sum_j A_{2,j} < C_f \left \Vert \mathbf{q}_0 \right \Vert
\right) \ge 
 1 - \frac{\mathds{E} \left[\sum_j A_{2,j}\right]}
{C_f \left \Vert \mathbf{q}_0 \right \Vert} \ge
1 - \frac{NT \sum_j \lambda_j}{C_f \left \Vert \mathbf{q}_0 \right \Vert},
\end{aligned}
\end{equation}
by the application of Markov's inequality for the random variable
$\sum_j A_{2,j}$.
}
\change{
Then we also change Equation (\ref{eq-Asqr}). It is easy to see that under the batch arrival model
\begin{equation} \label{eq-Asqr-batch}
\mathds{E}_{\mathbf{S}(t)}[A_j(t, t+u)^2] = 
\lambda \sum_{\mathbf{v} \in \mathcal{V}} v_j^2 p_{\mathbf{v}} u + \oo(u)
\end{equation}
}

\change{
Eventually this last result will change the expression of $B_2$ in
equation \ref{eq-AV-limit}, with $\rho_j$ being replaced by
$\frac{\lambda \sum_{\mathbf{v} \in \mathcal{V}} v_j^2 p_{\mathbf{v}}}{\mu_j}$.
}

\change{
Lastly we will have to update the constants of Lemma \ref{prop:weight_const}
to consider that the maximum change in number of jobs can be more than $K_{max}$
but is again bounded, since arrivals in each arrival event were assumed bounded.
}

\subsection{Proof of Lemma \ref{lem2}}\label{proof-lem2}
Let us first denote the normalized vector of resources of job type
$j$ as $\mathbf{w}_j = (w_{j1}, w_{j2}, \cdots, w_{jR})$ which means that the
values are normalized with the capacity of the server.
\change{
	Let $j^\prime$ be the job type which has the resource
 	with the highest relative value, i.e.,
	$j^\prime =\arg\max_{j \in \calJ}\left(Q_{j}(t)/(\max_n w_{jn})\right)$.
	We show that the maximal configuration that included only jobs of type
	${j^\prime}$ is $r$-max weight with $r=\frac{N_f}{R(N_f+1)}$.
}
This implies the configuration of job type $j=j^\star$ that maximizes
$Q_j(t) \left \lfloor 1/\max_{n=1,\cdots,R}{w_{jn}} \right \rfloor$
should also be $r$-max weight since its weight is greater than or equal to
that of $j^\prime$.

Using the job type $j^\prime$, the total number of jobs that can fit in the server is $\left \lfloor 1/\max_{n=1,\cdots, R}{w_{{j^\prime}n}} \right \rfloor$
jobs and the corresponding weight will be:
\begin{equation}
\begin{aligned}
f({\mathbf{k}^{(r)}}^\ell(t), \mathbf{Q}(t)) &= Q_{j^\prime}(t) \left 
\lfloor 1/\max_{n=1,\cdots,R}{w_{{j^\prime}n}} \right \rfloor
> \frac{N_f}{N_f+1} Q_{j^\prime}(t)/\max_n{w_{{j^\prime}n}} \\
&= \frac{N_f}{(N_f+1)R} Q_{j^\prime}(t)R/\max_n{w_{{j^\prime}n}}
\ge \frac{N_f}{(N_f+1)R} \max_{{\mathbf{k}}^\ell} f({\mathbf{k}}^\ell, \mathbf{Q}(t)),
\end{aligned} \nonumber
\end{equation}
where the last inequality follows because $Q_{j^\prime}(t)R/\max_n{w_{{j^\prime}n}}$
is equivalent with filling all $R$ resources with the maximum relative value job
$j^\prime$ without leaving residual capacity, which is an upper bound of the
max weight value $\max_{{\mathbf{k}}^\ell} f({\mathbf{k}}^\ell, \mathbf{Q}(t))$.

\subsection{Proof of Corollary~\ref{lem-beta}}\label{proof:lem-beta}
The term $\beta$ first appears in the proof of Theorem~\ref{theorem} in
Equation~(\ref{eq-integralk}) and is treated as constant. By focusing
on one term of that integral we will show how the bound will change
if $\beta$ is a function as defined in the previous description.
As a reminder
\begin{equation}
\begin{aligned}
&\mathds{E}_{\mathbf{S}(t_0)}\left[\int_{t=t_0}^{t_f}\sum_j Q_j(t)
\tilde{k}_j^\ell(t)\right] >
\mathds{E}_{\mathbf{S}(t_0)}\left[\int_{t=t_0}^{t_f}\sum_j Q_j(t) r\beta(\mathbf{Q}(t))
 {k^\star}_j^{\ell}(t)\right] > \\
&r\mathds{E}_{\mathbf{S}(t_0)}\left[\min_{t_0 \le t < t_f}{\beta(\mathbf{Q}(t))}\right]
\mathds{E}_{\mathbf{S}(t_0)}\left[\int_{t=t_0}^{t_f}\sum_j Q_j(t) {k^\star}_j^{\ell}(t)\right].
\end{aligned} \nonumber
\end{equation}

It then suffices to find a lower bound of $\mathds{E}_{\mathbf{S}(t_0)}[\min{\beta(\mathbf{Q}(t))}]$
for which we will prove that for large enough queues it is higher than
$(1-\epsilon)(1-\bar\epsilon) \bar\beta +\epsilon \bar\beta_{min}$
for any $\epsilon > 0$ and $\bar \epsilon > 0$.
Let value $\bar Q$ be such that, for any $\mathbf{Q}$ with
$\left \Vert  \mathbf{Q} \right \Vert_1 > \bar{Q}$,
$h(\mathbf{Q}) > (1 - \bar\epsilon) \bar\beta$ for some $\bar\epsilon > 0$.
Then
\begin{equation}
\begin{aligned}
&\mathds{E}_{\mathbf{S}(t_0)}\left[\min_{t_0 \le t \le t_f}
{\beta(\mathbf{Q}(t))}\right] > 
\mathds{P}\left(\min_{t_0 \le t \le t_f}
{\left \Vert \mathbf{Q}(t)\right \Vert_1} >\bar Q | \mathbf{S}(t_0)\right)
 (1-\bar\epsilon)\bar\beta +
\mathds{P}\left(\min_{t_0 \le t \le t_f}
{\left \Vert \mathbf{Q}(t)\right \Vert_1}
\le \bar Q | \mathbf{S}(t_0)\right) \beta_{min}.
\end{aligned} \nonumber
\end{equation}

The result follows if we can have
$\mathds{P}(\min{\left \Vert \mathbf{Q}(t)\right \Vert_1} >
\bar Q | \mathbf{S}(t_0)) > 1-\epsilon$.
Using the shorthand $\mathbf{Q}(t_0) = \mathbf{q}_0$ we have
\begin{equation}
\begin{aligned}
&\mathds{P}(\min_t{\left \Vert \mathbf{Q}(t)\right \Vert_1} > \bar Q | \mathbf{S}(t_0)) > 
\mathds{P}(\min_t{\left \Vert \mathbf{Q}(t)\right \Vert_1} >
	C \left \Vert \mathbf{q}_0\right \Vert_1 | \mathbf{S}(t_0)) \cdot
\mathds{1}(\left \Vert \mathbf{q}_0\right \Vert_1 > \bar Q/C)  \ge \\
&\mathds{P}(\min_t{\left \Vert \mathbf{Q}(t)\right \Vert} >
	\sqrt{J} C \left \Vert \mathbf{q}_0\right \Vert | \mathbf{S}(t_0)) \cdot
\mathds{1}(\left \Vert \mathbf{q}_0\right \Vert > \bar Q/C).
\end{aligned} \nonumber
\end{equation}

Finally assuming $\left \Vert \mathbf{q}_0\right \Vert > \bar Q/C$
and process of Equation~(\ref{p4}) we have
\begin{equation} \label{p-lem3}
\begin{aligned}
&\mathds{P}(\min_t{\left \Vert \mathbf{Q}(t)\right \Vert} >
	\sqrt{J} C \left \Vert \mathbf{q}_0\right \Vert) >
\left (1- \frac{\lambda_{max} NT}
	{\sqrt{J} C \left \Vert \mathbf{q}_0 \right \Vert}\right)^J
\left (1- \frac{K_{max} NT^2}
	{\sqrt{J} C \left \Vert \mathbf{q}_0 \right \Vert}\right)^J
	> 1 - \bar\epsilon,
\end{aligned}
\end{equation}
with the last inequality being true when
\begin{equation}
\left \Vert \mathbf{q_0} \right \Vert
  > \frac{2LJNT \max(\lambda_{max}, K_{max} T)}
  {\bar\epsilon \sqrt{J} C} \nonumber
\end{equation}

The last derivation follows the same steps as the one that led to
formula~(\ref{Mq-cond}). The condition~(\ref{p-lem3}) is satisfied
for all initial queue sizes except possibly for those for which
\begin{equation}
\left \Vert \mathbf{q_0} \right \Vert
  < \max \left(\frac{2LJNT \max(\lambda_{max}, K_{max} T)}
  {\bar\epsilon \sqrt{J} C}, \frac{\bar Q}{C} \right). \nonumber
\end{equation}

\end{document}